%% file: omni_V11.tex
\documentclass[11pt,draftclsnofoot,peerreviewca,onecolumn,letterpaper]{IEEEtran}
\usepackage{amsmath}
\usepackage{amssymb,theorem}
\usepackage[pdftex]{graphics}
\usepackage{epsfig}
\usepackage[pdftex]{color}
\usepackage{psfrag,bbm}
\usepackage{relsize}

\usepackage{latexsym}
\usepackage{epsfig}

\usepackage{graphicx}
\usepackage[font=small]{caption}
\usepackage[font=footnotesize]{subcaption}
\usepackage{wrapfig}
\usepackage{cite}

\newtheorem{thm}{Theorem}

\newtheorem{lemma}{Lemma}
\newtheorem{cor}{Corollary}

\newtheorem{defn}{Definition}
\newtheorem{remark}{Remark}

\input{macros}

\IEEEoverridecommandlockouts

\begin{document}

\title{Cellular Interference Alignment:  Omni-Directional Antennas and Asymmetric Configurations}

\vspace{-0.2in}
\author{Vasilis~Ntranos$^\dagger$\thanks{This work is the outcome of a collaboration that started while V. Ntranos was a research intern at Bell Labs, Alcatel-Lucent. Emails: ntranos@usc.edu,
mohammadali.maddah-ali@alcatel-lucent.com, caire@usc.edu.}, 
        Mohammad~Ali~Maddah-Ali$^\ast$,  and
        Giuseppe~Caire$^\dagger$ \\
        $^\dagger$University of Southern California, Los Angeles, CA,  USA\\ 
				$^\ast$Bell Labs, Alcatel-Lucent, Holmdel, NJ, USA }
        
\maketitle

\vspace{-12pt}

\begin{abstract}
Although interference alignment (IA) can theoretically achieve the optimal degrees of freedom (DoFs) in the $K$-user Gaussian interference channel, 
its direct application  comes at the prohibitive cost of precoding over exponentially-many signaling dimensions. 
On the other hand, it is known that practical ``one-shot'' IA precoding (i.e., linear schemes without symbol expansion) 
provides a vanishing DoFs gain in large fully-connected networks with generic channel coefficients. 
In our previous work, we introduced the concept of ``Cellular IA'' for a network topology induced by  hexagonal cells with sectors and 
nearest-neighbor interference. Assuming that neighboring sectors can exchange decoded messages (and not received signal samples) 
in the uplink, we showed that linear one-shot IA precoding over $M$ transmit/receive antennas can achieve 
the optimal $M/2$ DoFs per user.
In this paper we extend this framework to networks with omni-directional (non-sectorized) 
cells and consider the practical scenario where users 
have $2$ antennas,  and  base-stations have $2$, $3$ or $4$ antennas. 
In particular, we provide linear one-shot IA schemes
for  the $2\times 2$, $2\times3$ and $2\times 4$ cases,  and show the achievability of  $3/4$, $1$ and $7/6$ DoFs per user, respectively.  
DoFs converses for one-shot schemes require the solution of a discrete optimization problem over a number of variables that grows 
with the network size. We develop a new approach to transform such challenging optimization problem into a tractable linear program (LP) 
with significantly fewer variables. This approach is used to show that the achievable $3/4$ DoFs per user are indeed optimal 
for a large (extended) cellular network with $2\times 2$ links.

\end{abstract}

\vspace{-12pt}


\section{Introduction}
Interference management is arguably one of the most important technical challenges in the design of  wireless systems. Conventional techniques, such as time/frequency orthogonalization, waste valuable channel resources in order to avoid interfering transmissions (by giving each user a unique fraction of the spectrum) and soon will not be  able to  keep up with the rapidly increasing bandwidth demands of today's networks. 
Recent  advances in information theory  \cite{cj08,mgmk09,ergodic}  have shown that
transmission schemes based on {\it interference alignment} 
\cite{mmk08,cj08} are  able to provide half of the available spectrum to each user in the network and  promise significant gains compared to conventional  approaches. However, the extent to which such gains 
can be realized in practice has been so far limited.

In order to achieve the theoretically optimal performance, interference alignment solutions often rely on    infinite channel diversity 
and asymptotic symbol expansion \cite{cj08,sy13}, or infinite-layer lattice alignment \cite{mgmk09}, rendering their direct applicability to wireless systems virtually impossible.
An emerging body of work has focused on practical settings, investigating the performance of interference alignment at finite SNR and restricting the achievability to linear beamforming schemes without symbol extensions. However, in the majority of cases, the resulting solutions do not scale accordingly as the size of the network increases and the corresponding gains have only been materialized in settings with a small number of users  (e.g., three or four transmit-receive pairs) \cite{mmk08,sht11,katabi09,sy13isit}. 
In fact, it has been shown that without symbol extensions, the degrees of freedom (DoF) gain of any linear interference alignment scheme  in a 
fully-connected network {with generic channel coefficients}\footnote{For example, coefficients independently sampled from a continuous distribution.} 
vanishes as the number of users increases~\cite{ygjk10, Razaviyayn, Bresler}. 
These results raise an important question: \emph{Is there any reasonable cellular system deployment in which interference alignment is ``practically'' feasible and yields sizable DoFs gain with respect to trivial time-frequency orthogonalization?} Motivated by this question,
in our previous work \cite{nmc14} we introduced a novel framework for the uplink
of 
extended\footnote{Following \cite{extendedkumar04,extendedtelatar05,extendedtse10} 
we refer to an ``extended'' network as a network with a fixed spatial density of cells and increasing total coverage area, in 
contrast to a ``dense'' network where the total coverage area is fixed and the cell density increases.}
sectored cellular networks in which interference alignment can achieve
 the promised optimal DoFs by linear precoding in one-shot (i.e., by precoding over a single time-frequency slot). 
 In particular,  we considered a  backhaul network architecture, in which nearby receivers can exchange already decoded messages and showed that this {\it local} and {\it one directional} data exchange  is enough  to reduce the uplink of a sectored cellular network to a topology in which the optimal degrees 
of freedom can be achieved by linear interference alignment schemes without requiring time-frequency expansion or lattice alignment.  
In contrast to existing works on ``Network MIMO'' \cite{kfv06,fkv06,multicell10,htc12} and the popular and widely studied ``Wyner model'' \cite{wyner94,shamaiwyner97} for cellular systems, our proposed architecture does not require
that the base-station receivers share received signal samples and/or the presence of a central processor 
performing joint decoding of multiple user messages, which is arguably much more demanding for 
the backhaul connections. The idea of combining IA with network interference cancellation has also been considered in \cite{katabi09} for small network configurations (e.g, with three active receivers in the uplink). In contrast, our ``Cellular IA'' framework applies to large (extended) networks in which interference alignment without asymptotic symbol expansion was not known to be feasible.


In this paper we follow a similar approach and focus on the uplink of {\it non-sectored cellular systems} where base-stations are equipped with omni-directional antennas. 
The motivation behind this work is both practical and theoretical. From the practical viewpoint, 
omni-directional base-stations are typically used in dense small-cell deployments \cite{andrews2012femtocells} where intercell
interference from neighboring cells is a major impairment. From the theoretical viewpoint, 
the wireless system scenario considered here is much more challenging than the sectored case considered in \cite{nmc14}, 
since base-station receivers operate in a richer interference environment: In  \cite{nmc14}, each sector receiver observed   four dominant interfering links from its neighboring out-of-cell sector transmitters due to the deployment of directional antennas, whereas here it is natural to assume that every base-station receiver in the network will observe significant interference from all six surrounding neighboring cells.
Under this framework, we  focus on   cellular system configurations in which user terminals are equipped with $M=2$ transmit antennas and provide one-shot linear interference alignment schemes for the  cases in which base-stations  are equipped with $N=2$, $3$ and $4$ receive antennas. Even though it is straightforward to extend our achievable schemes  to  the $M\times M$ case (as in \cite{nmc14}), we chose here to take a different approach and investigate the above \emph{asymmetric} antenna configurations that are more practical and relevant to current cellular system deployments.  
Further, we propose a new converse technique for  cellular networks with $M\times M$ links (i.e., transmitters and receivers have the same number $M$ of antennas) based on the DoFs feasibility inequalities for linear interference alignment introduced in \cite{Razaviyayn},\cite{Bresler}. Direct application of these inequalities in our setting yields a challenging  (non-linear, discrete optimization) problem over a number of variables that grows with the number of transmit-receive pairs in the network. To overcome this difficulty, we exploit the topology of the interference graph and reformulate the corresponding optimization problem into a linear program (LP) with a small number of variables that does not depend on the size of the network.  Using this approach, we are able to provide a tight outer bound on the achievable DoFs in the $2\times 2$ case for large networks and show that our one-shot IA scheme is optimal.
%

This paper is organized as follows. First, in Section \ref{cellmodel} we describe the cellular model that we consider in this work and give a formal problem statement. Then, in Section \ref{sec:2x2}  we state our results for $2\times 2$ cellular networks and give the corresponding achievability and converse theorems. Finally,
in Section \ref{sec:asym} we consider networks with asymmetric antenna configurations and provide the corresponding interference alignment schemes for $2\times 3$ and $2\times 4$ systems.


\section{Cellular Model}\label{cellmodel}

We consider a large MIMO cellular network with $K$ base-stations, each one serving a user's mobile terminal as depicted in Figure~\ref{cell}. Within each cell, the base-station is  
interested in decoding the uplink transmissions of the user associated with it \begin{figure*}[ht]
        \centering
        \begin{subfigure}[b]{0.49\columnwidth}
                \centering
                \includegraphics[width=\columnwidth]{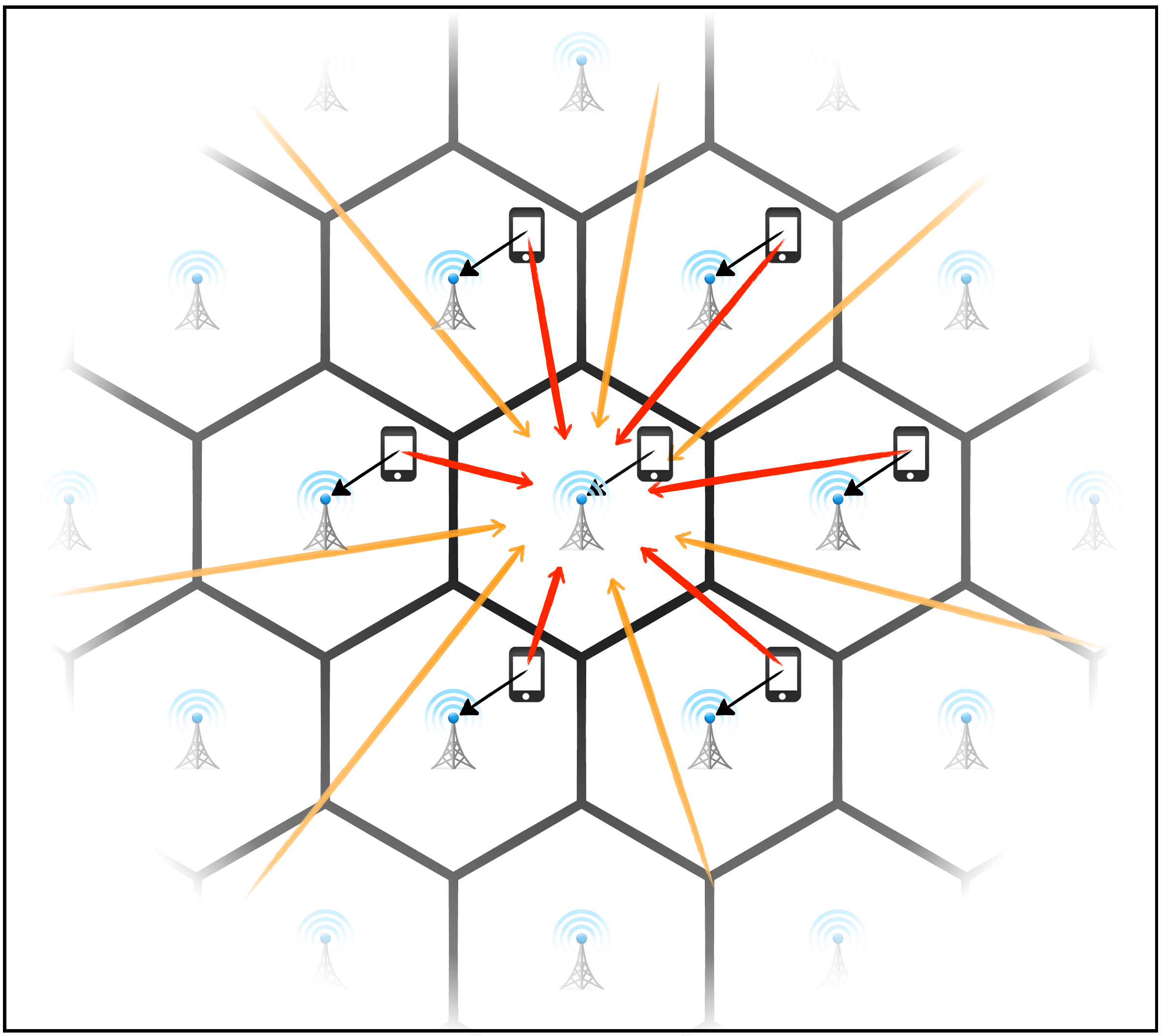}
                \caption{Cellular network}
\label{cell}
        \end{subfigure}%
        ~       \begin{subfigure}[b]{0.49\columnwidth}
                \centering
                \includegraphics[width=\columnwidth]{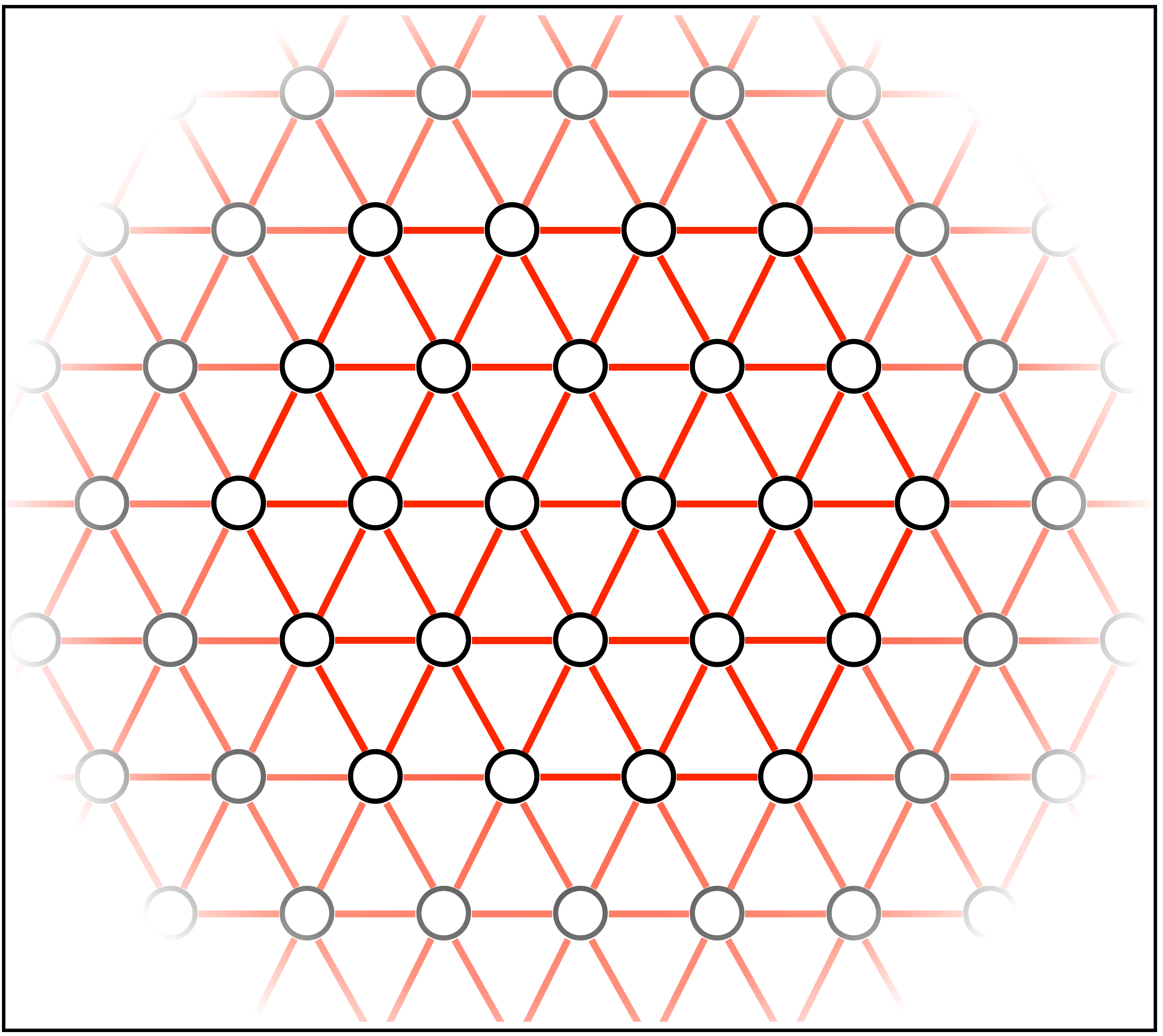}
\caption{Interference graph}
\label{igraph}
        \end{subfigure}
        
       \caption{ The cellular network topology and the corresponding dominant interference graph. The desired uplink channels from mobile terminals to their associated base-stations are depicted with black arrows in each cell. In practice, each base-station observes six dominant interfering signals from its six adjacent cells (red arrows) and 
{\it weaker interference} from outer cells (orange arrows) located at distance two or more. 
In our model, we consider the interference graph shown in (b) which captures only the dominant sources of interference for each cell; the vertices represent transmit-receive pairs  and edges indicate interfering neighbors.}\label{fig:1}
\end{figure*}and observes all other simultaneous transmissions as interference.  
We assume that transmitters and    receivers     are equipped with $M$ and $N$ antennas, respectively, and consider flat fading channel gains that remain constant throughout the entire communication.

Taking into account path loss and shadowing effects that are inherent to wireless transmissions, 
we assume that all interference in our cellular model is generated {\it locally} between transmitters and receivers of neighboring cells. As depicted in Fig.~1a, each base-station receiver in the network will observe dominant interference from all of its six neighboring cell transmitters and therefore the cellular network can be modeled as a partially-connected interference channel with the hexagonal topology shown in Fig.~1b.

Let $\mathcal S$ be the index set of all cells in the network and let $\mathcal{N}(i)$ denote the six interfering neighbors of the cell $i\in\mathcal{S}$. Within our framework, the received observation of the $ith$ base-station can be written as   
\begin{equation}
\yv_{i} = \Hm_{ii}\xv_{i} + \sum_{j\in {\mathcal N}(i)}\Hm_{ij}\xv_{j} + \zv_{i} 
\end{equation}
where $\Hm_{ij}$ is the $N\times M$ matrix of channel gains between the transmitter  associated with cell $j$ and the receiver of cell $i$ and  $\xv_{i}$ are the corresponding transmitted signals satisfying the average power constraint $\mathbb{E}\big[||\xv_{i}||^{2}\big]\leq P$.

\subsection{Interference Graph} \label{sec:igraph}

A useful representation of our cellular model can be given by the corresponding {\it interference graph} ${\cal G}({\cal V},{\cal E})$ shown in Fig.~1b. In this graph vertices represent transmit-receive pairs within each cell and 
edges indicate interfering neighboring links: the transmitter associated with a node $u\in \cal V$ causes interference to all receivers associated with nodes $v\in \cal V$ if there is an edge $(u,v)\in \cal E$.  
Notice that  the interference graph is undirected and hence interference between cells in our model goes in both directions. 

It is convenient to represent the interference graph $\cal G(V,E)$ by identifying $\cal V$ with a set of points on the complex plane whose coordinates are referred to as the node labels. The geometry of such labels in the complex plane corresponds to the hexagonal lattice layout of the cells as shown in Fig.~1. A natural choice for this labeling that we will use throughout this paper is the set of the Eisenstein integers $\mathbb{Z}(\omega)$ shown in Fig.~\ref{eisen}.


\begin{defn}[Eisenstein integers]
The Eisenstein integers denoted as $\mathbb{Z}(\omega)$, are defined as the set of complex numbers of the form $z=a+b\omega$, where $a,b \in \mathbb{Z}$ and $\omega = \frac{1}{2}(-1 + i\sqrt{3})$. 
\end{defn}

\begin{figure}[h]
                \centering
                \includegraphics[width=0.65\columnwidth]{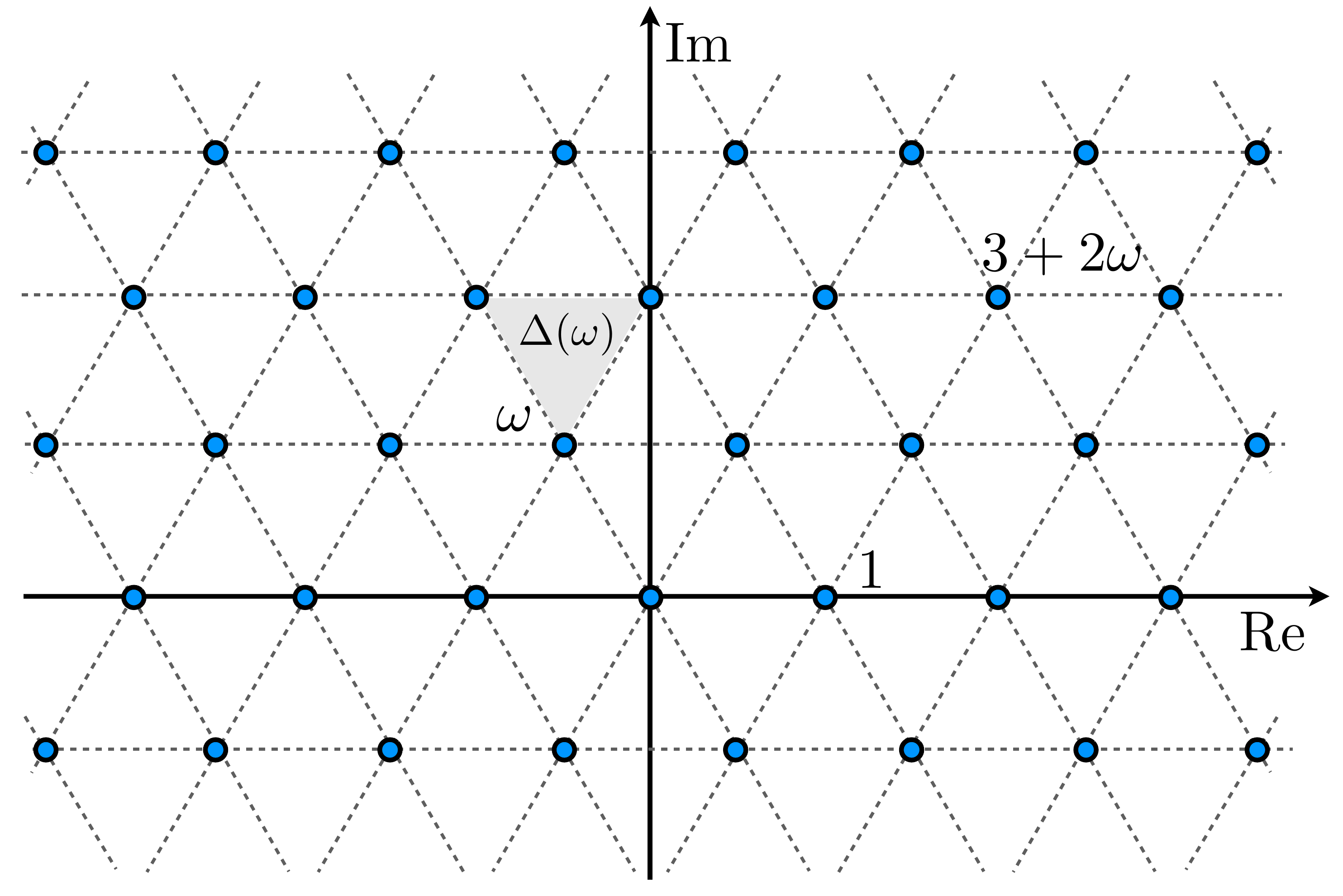}
                \caption{ The Eisenstein integers $Z(\omega)$ on the complex plane. }
                \label{eisen}

\end{figure}

Define ${\cal B}_{r} \triangleq \{z\hspace{-0.05in }\in \hspace{-0.05in }\mathbb{C}\hspace{-0.05in } :\hspace{-0.02in } |{\rm Re}(z)|\leq r , |{\rm Im}(z)|\leq \frac{\sqrt{3}r}{2}\}$
and let $\phi: {\cal V} \rightarrow \mathbb{Z}(\omega)\cap{\cal B}_{r}$ be a one-to-one mapping between the elements of $\cal V$ and the set of bounded Eisenstein integers given by $\mathbb{Z}(\omega)\cap{\cal B}_{r}$. For any $v\in \cal V$ we say that $\phi(v)$ is the unique label of the corresponding node in our graph.  
Correspondingly, the set of vertices $\cal V$ is given by
\begin{equation}
{\cal V} = \left\{ \phi^{-1}(z) : z\in \mathbb{Z}(\omega)\cap{\cal B}_{r}\right\}.
\label{eq:V}
\end{equation}
In order to explicitly describe the set of edges $\cal E$ in terms of the function $\phi$, we define the set
\begin{equation}
{\cal D}\triangleq \underset{{ z\in\mathbb{Z}(\omega)}}{\bigcup}{\Delta}(z),
\vspace{-0.15in}
\end{equation} 
where
\begin{equation}
{\Delta}(z) = \{(z,z+\omega),\,(z,z+\omega+1),\,(z+\omega,z+\omega+1) \}
\end{equation}
is the set of the three line segments in $\CC$ (shown in Fig.~\ref{eisen}) that form a triangle with vertices $z$, $z+\omega$ and $z+\omega+1$.
The set of edges ${\cal E}$ in our graph can hence be given by
\begin{equation}
{\cal E} = \left\{(u,v) : u,v\in{\cal V} \mbox{ and} \left(\phi(u),\phi(v)\right)\in {\cal D}\right\}.
\label{eq:E}
\end{equation}

\begin{defn}[Interference Graph]
The interference graph $\cal G (\cal V, \cal E)$ is an undirected graph defined by the set of vertices $\cal V$ given in (\ref{eq:V}) and the corresponding set of edges $\cal E$ given in (\ref{eq:E}). The graph vertices represent transmit-receive pairs in our cellular model and edges indicate interfering neighbors.
\end{defn}




\subsection{Network Interference Cancellation}\label{sec:NICE}
We further consider a message-passing network architecture for our cellular system, in which  base-station receivers  communicate locally in order  to exchange  decoded messages.
Any  receiver that has already decoded its own user's message can use the backhaul of the network and  pass it as side information to one or more of its neighbors. 
In turn, the neighboring base-stations can use the received decoded messages in order to  reconstruct the corresponding interfering  
signals and subtract them from their observation. It is important to note that this scheme only requires sharing (decoded) information messages between neighboring  receivers and does not require sharing the baseband signal samples, which is much more demanding for the backbone network.
%

The above operation effectively cancels  interference in one direction: 
all decoded messages propagate through the backhaul of the network, successively eliminating certain interfering 
links between neighboring base-stations according to a specified decoding order.  
Fig.~\ref{NICE} illustrates the above network interference cancellation process in our cellular graph model assuming a ``left-to-right, top-down'' 
decoding order.  Notice that edges are now {\it directed} in order to indicate the interference flow over the network. 
For example, if an undirected edge $(u,v)$ exists in $\Ec$ and, under this message-passing architecture, 
node $v$ decodes its message before node $u$ and passes it to node $u$ through the backhaul, 
then the resulting interference graph will contain the directed link $[u,v]$, indicating that the interference is from node   $u$ to
node  $v$ only.

\begin{figure}[ht]
                \centering
                \includegraphics[width=0.78\columnwidth]{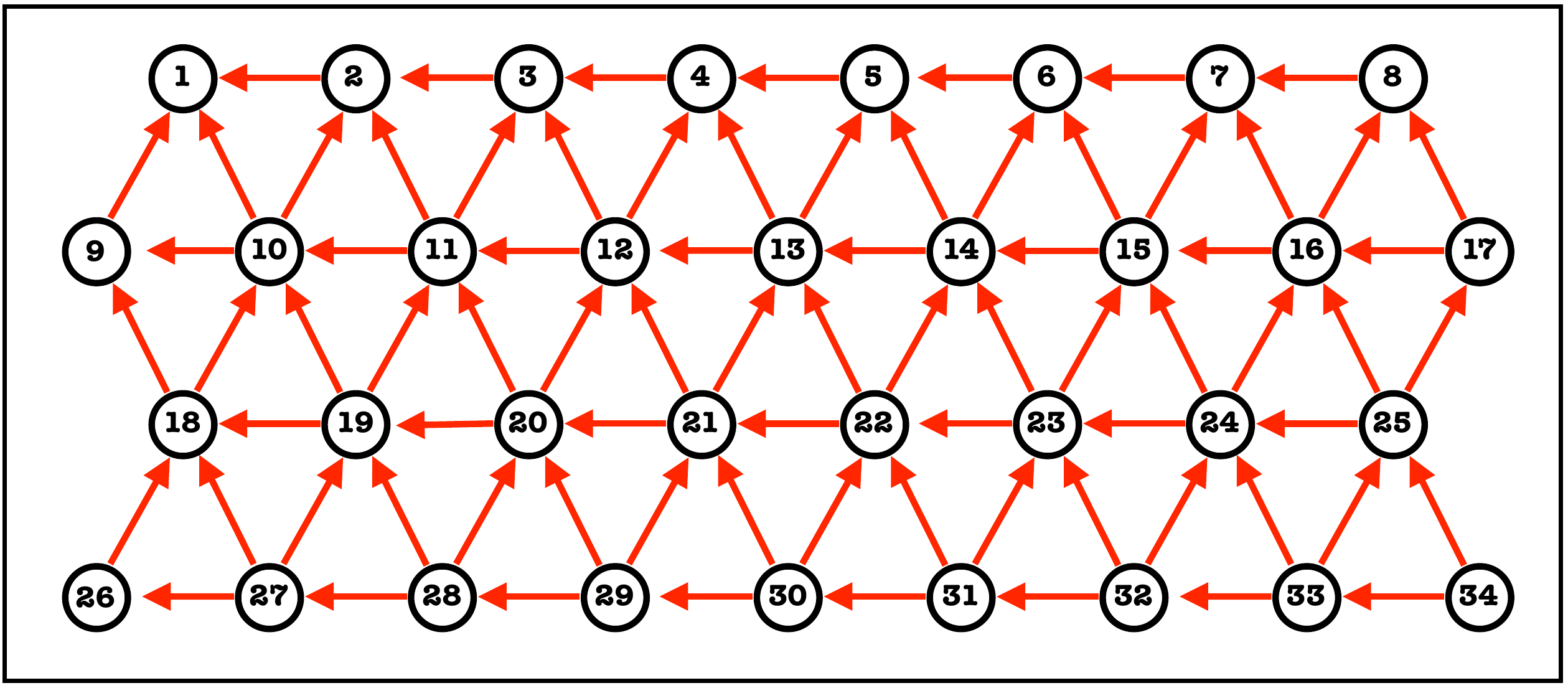}
                \caption{ The directed interference graph $\cal G_{\pi^{*}}(V,E_{\pi^{*}})$ after network interference cancellation according to the ``left-to-right, top-down'' decoding order $\pi^{*}$. The transmitter of a cell associated with node $i$ causes interference only to its neighboring base-station receivers $j$ with $j<i$.  \vspace{-0.1in}}
                \label{NICE}

\end{figure}

A decoding order $\pi$ can be specified by defining a 
partial order ``$\prec_{\pi}$'' over the  set of vertices $\cal V$ in our interference graph.  Then, the message of the user associated with vertex $v \in \cal V$ will be decoded before the one associated with vertex $u \in \cal V$ if $v \prec_{\pi}u$. In principle, we can choose any decoding order that partially orders the set $\cal V$ and hence $\pi$ can be treated as an optimization parameter in our model. 

\begin{defn}
[Directed Interference Graph ${\cal G}_{\pi}$] For a given partial order ``$\prec_{\pi}$'' on $\cal V$, the directed interference graph is defined as ${\cal G}_{\pi}({\cal V}, {\cal E}_{\pi})$ 
where 
${\cal E}_{\pi}$ is a set of ordered pairs $[u,v]$ given by 
$
{\cal E}_{\pi} = \left\{[u,v] : (u,v) \in {\cal E} \mbox{ and }  v\prec_{\pi}u   \right\}
$. \hfill $\lozenge$
\end{defn}

Next, we formally specify the ``left-to-right, top-down'' decoding order $\pi^{*}$ that has been chosen in Fig.~\ref{NICE} and will be used for the rest of this paper. As we will show in the following section,  this decoding order is indeed optimum for large cellular networks with $M=N=2$ and it can lead to the maximum possible DoF per user under our framework.

\begin{defn}[The Decoding Order $\pi^{*}$]
The ``left-to-right, top-down'' decoding order $\pi^{*}$ is defined by the partial ordering $\prec_{\pi^{*}}$ over $\cal V$ such that for any $u, v\in \cal V$, $v\prec_{\pi^{*}}u \Leftrightarrow$
\begin{equation*}
 \begin{cases} {\rm Im}\left(\phi(v)\right) > {\rm Im}\left(\phi(u)\right) \mbox{, or} \\ 
{\rm Im}\left(\phi(v)\right) = {\rm Im}\left(\phi(u)\right) \mbox{and}\;  {\rm Re}\left(\phi(v)\right) < {\rm Re}\left(\phi(u)\right)
\end{cases}
\end{equation*}
\hfill $\lozenge$

\end{defn}

\subsection{Problem Statement}\label{sec:probstate}

Our main goal is to design efficient communication schemes for the cellular model  introduced in this section.
As a first-order approximation of a scheme's efficiency, we will consider here the achievable DoFs, broadly defined as the number of point-to-point interference-free channels that can be created between transmit-receive pairs in the network.

More specifically, we are going to limit ourselves to linear beamforming strategies over multiple antennas assuming constant (frequency-flat) 
channel gains, independently chosen from a continuous non-degenerate distribution, without allowing symbol extensions. We refer to such schemes as ``one-shot'', indicating that precoding is achieved over a single
time-frequency slot (symbol-by-symbol).  Our goal it to maximize, over all decoding orders $\pi$, the average (per cell) achievable DoFs

\begin{equation}
d_{{\cal G},{\pi}} \triangleq \frac{1}{|{\cal V}|} \sum_{v\in \cal V}d_{v}\, ,
\end{equation}
where  $\cal G(V,E)$ is the  interference graph defined in Section \ref{sec:igraph} and 
$d_{v}$ denotes the DoFs achieved by the transmit-receive pair associated with the node 
$v\in\cal V$, where 
$$d_{v} = \lim_{P\rightarrow \infty}\frac{R_{v}(P)}{\log(P)},$$ 
and $R_{v}(P)$ is  the achievable rate in cell 
$v\in \cal V$ under the per-user transmit power constraint $P$.

\section{Cellular Interference Alignment: $M=2, N=2$}\label{sec:2x2}

\begin{thm}[Achievability] \label{thm1}
For a $2\times2$ cellular system $\cal G(V,E)$, there exist a 
one-shot linear beamforming scheme that achieves the average (per cell) DoFs,
$d_{{\cal G},{\pi^{*}}} = {3}/{4}$,
under the network interference cancellation framework with decoding order $\pi^{*}$.
\hfill \QED
\end{thm}

\begin{thm}[Converse]  \label{thm2}
For a $2\times2$ cellular system $\cal G(V,E)$  the average (per cell) DoFs $d_{{\cal G},{\pi}}$  that can be achieved by any one-shot 
linear beamforming scheme, for any network interference 
cancellation decoding order $\pi$,  are bounded by 
$\textstyle
 d_{{\cal G},{\pi}} \leq 3/4 + {\cal O}\left( \scriptstyle{1}/{{\sqrt{|{\cal V}|}}}\right)
$.
\hfill \QED
\end{thm}

\begin{remark}\rm
The above results can be directly translated to achievability and converse theorems for sectored cellular systems with intra-cell interference. 
In \cite{nmc14} we studied a wireless network scenario under a similar framework in which co-located sectors (i.e., the  sectors of the same cell) are able to jointly process their received observations and showed that for $M\times M$ links the optimal $M/2$ DoFs (per user) are achievable. An interesting observation is that the cellular model considered in this paper 
leads to an interference graph that is similar (isomorphic) to the one considered in the above case. 
The fundamental difference here is that we do not allow any receivers to jointly process their signals. 
Therefore, it is possible to restate the results of Theorems \ref{thm1} and \ref{thm2} for the sectored case and assess the gain of joint cell processing in such systems. When $M=2$, $N=2$, we can see that jointly processing the received signals within each cell yields 33\% gain in terms of the average achievable  DoFs, compared to single-user decoding. \hfill \QED
\end{remark}

\subsection{Achievability for $M=2$, $N=2$. (Proof of Theorem \ref{thm1})}

Consider the  interference graph ${\cal G}({\cal V}, {\cal E})$ introduced in Section \ref{sec:igraph} and assume that all  user terminals $v\in \cal V$ whose labels $\phi(v)$ belong to the sub-lattice  
\begin{equation}
\Lambda_{0}\triangleq 2\cdot\ZZ(\omega) \label{L0}
\end{equation}  
are turned off, while the remaining user terminals with $\phi(v)\in \ZZ(\omega)\setminus \Lambda_{0}$ are 
all simultaneously transmitting their signals  $\xv_{v}$ to their corresponding receivers. 
Let ${\cal V}_{0}\triangleq \{v\in {\cal V}: \phi(v)\in \Lambda_{0}\}$ denote the set of inactive vertices   and let ${\cal E}_{0}\triangleq \{(u,v)\in {\cal E}: \phi(u) \mbox{ or } \phi(v) \in \Lambda_{0}\}$ be the set of edges that are adjacent 
to ${\cal V}_{0}$.


Recall that under our framework, each base-station receiver that is able to decode its own message, is also able to pass it as side information to its neighbors, effectively eliminating interference in that direction. 
Hence, following  the ``left-to-right, top-down'' decoding order $\pi^{*}$ introduced in Section \ref{sec:NICE}, the  receivers associated with the active nodes $u\in {\cal V}\setminus {\cal V}_{0}$ are able to  eliminate  interference from  all neighboring cells $v\prec_{\pi^{*}}u$ 
and attempt to decode their own message from the two-dimensional received signal observation $\yv_{u}$ given by
\begin{equation}
\yv_{u} = \Hm_{uu}\xv_{u} + \sum_{v : [v,u]\in {\cal E}_{\pi^{*}}\setminus {\cal E}_{0}} \Hm_{uv}\xv_{v} +\zv_{u}.
\end{equation}


Our goal is to design the transmitted signals $\xv_{v}$ such that  all interference observed in  
$\yv_{u}$ is aligned in one dimension for all $u\in {\cal V}\setminus {\cal V}_{0}$. In that way, we can show that the DoFs 
$$d_{v} = \begin{cases} 1, \;\;v \in {\cal V}\setminus {\cal V}_{0}\\ 0, \;\;v\in {\cal V}_{0}\end{cases}$$
are achievable in $\cal G(V,E)$ and hence 
\begin{equation*}
d_{{\cal G},{\pi}^{*}} = \frac{1}{|{\cal V}|}\sum_{v\in \cal V}d_{v} = \frac{|{\cal V}\setminus {\cal V}_{0}|}{|{\cal V}|}=1-\frac{|{\cal V}_{0}|}{|{\cal V}|}.
\end{equation*}
Later, we will see that $\Lambda_{0}$ has been chosen so that $|{\cal V}_{0}|\leq \frac{1}{4}|{\cal V}|$, and hence obtain that the average (per cell) DoFs $d_{{\cal G},{\pi}^{*}} ={3}/{4}$ are indeed achievable under our framework.

\begin{figure}[h]
                \centering
                \includegraphics[width=0.55\columnwidth]{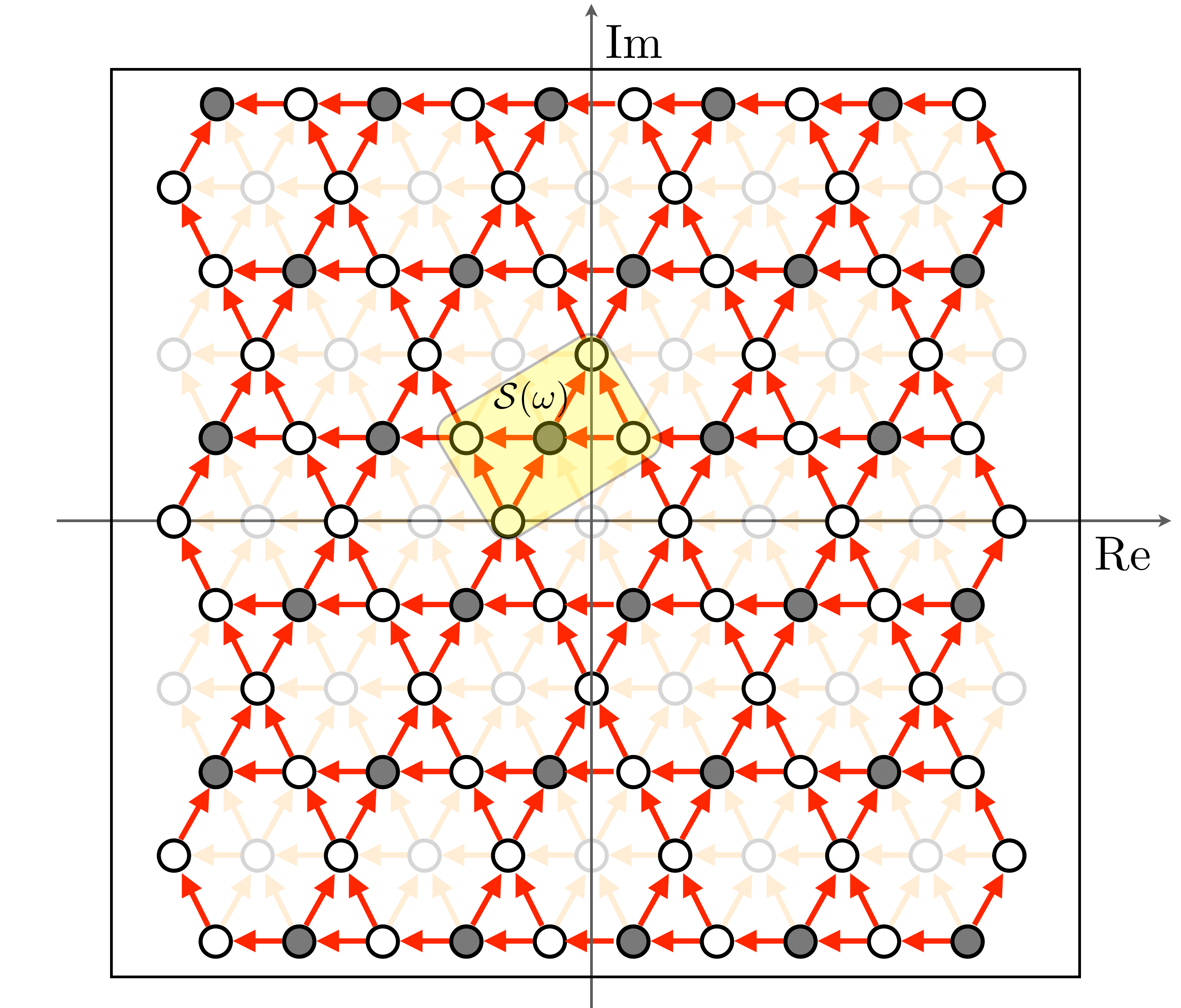}
                \caption{ The  interference sub-graph 
                ${\hat{\cal G}}_{\pi^{*}}({\cal V} \setminus {\cal V}_{0},{\cal E}_{\pi^{*}}\setminus {\cal E}_{0})$   represented on the   complex plane. The transparent vertices correspond to  inactive cells $v\in {\cal V}_{0}$ with labels $\phi(v) \in \Lambda_{0}$ and the transparent edges correspond to their adjacent edges ${\cal E}_{0}$. 
                The  edges in the cluster  ${\cal S}(\omega)$ are highlighted in a tilted rectangle that is centered at the point $z=\omega$. Notice that all the interfering edges in the above graph can be partitioned into smaller (isomorphic) sets  by translating  ${\cal S}(z)$ over all $z\in \Lambda_{0}+\omega$ (black vertices).
                }
                \label{butterfly}
\end{figure}

The  setting described above is illustrated in Fig.~\ref{butterfly}, where transparent vertices correspond to  inactive cells $v\in {\cal V}_{0}$ and transparent edges correspond to ${\cal E}_{0}$. The resulting interference graph, denoted  here as ${\hat{\cal G}_{\pi^{*}}}\left({\cal V} \setminus {\cal V}_{0},{\cal E}_{\pi^{*}}\setminus {\cal E}_{0}\right)$, is a sub-graph of ${\cal G}_{\pi^{*}}({\cal V},{\cal E}_{\pi^{*}})$ and represents the corresponding interference between  active transmit-receive pairs in the network.

Let $\uv_{u}$ and $\vv_{u}$ denote the 2-dimensional receive and transmit beamforming vectors associated with the active nodes $u\in  {\cal V} \setminus {\cal V}_{0}$ and assume that the active user terminals in the network have encoded their messages in the corresponding codewords. 
Although codewords span many slots (in time), we focus here on a single slot and denote the corresponding coded symbol of user $u$ by 
$s_{u}$. Then, the vector transmitted by  user $u$ 
is given by $\xv_{u}=\vv_{u}s_{u}$ and each receiver can project its observation $\yv_{u}$ along $\uv_{u}$ to obtain

$$\hat y_{u} = \uv_{u}^{\rm H}\Hm_{uu}\vv_{u}s_{u} 
+ \sum_{v : [v,u]\in {\cal E}_{\pi^{*}}\setminus {\cal E}_{0}} \uv_{u}^{\rm H}\Hm_{uv}\vv_{v}s_{v} + \hat z_{u}.
$$

We will show next that it is possible to design $\uv_{u}$ and $\vv_{u}$ across the entire network  such that the following interference alignment conditions are satisfied:
\begin{eqnarray}
&\uv_{u}^{\rm H}\Hm_{uu}\vv_{u} \neq 0,&\; \forall u\in {\cal V}\setminus {\cal V}_{0}  \;\;\; \mbox{and}
\label{iacond10}\\
&\uv_{u}^{\rm H}\Hm_{uv}\vv_{v} =0,& \; \forall [v,u]\in {\cal E}_{\pi^{*}}\setminus {\cal E}_{0}.\label{iacond11}
\end{eqnarray}
Hence, every active receiver in the network can decode its own desired symbol $s_{u}$ from an interference-free channel observation of the form \begin{equation} \hat y_{u} = \hat h_{u} s_{u} + \hat z_{u}\label{form}\end{equation}
where $\hat h_{u} = \uv_{u}^{\rm H}\Hm_{uu}\vv_{u}$ and $\hat z_{u} = \uv_{u}^{\rm H}\zv_{u}$.

In order to describe the alignment precoding scheme, we will partition the interfering  edges in ${\hat{\cal G}}_{\pi^{*}}({\cal V} \setminus {\cal V}_{0},{\cal E}_{\pi^{*}}\setminus {\cal E}_{0})$ into smaller sets that we will refer to as the interference clusters, given by
\begin{align}
{\cal S}(z)\triangleq \big\{&[z,z-1],[z,z+1+\omega],\\&[z+1,z],[z+1,z+1+\omega], \\&[z-1-\omega, z-1], [z-1-\omega, z] \big\}, 
\end{align}
for  $z\in \Lambda_{0} + \omega$. To illustrate the above definition, in Fig.~\ref{butterfly}, the interference cluster ${\cal S}(\omega)$ is highlighted in a tilted rectangle that is centered at $z=\omega$ and the points $z\in \Lambda_{0} + \omega$ are shown in black color.
It is not hard to verify that the set of edges can be written as
\begin{align}
{\cal E}_{\pi^{*}}\setminus {\cal E}_{0} &= \bigcup_{z\in \Lambda_{0} + \omega}\{[u,v]: [\phi(u),\phi(v)]\in {\cal S}(z), u,v\in {\cal V}\} \\
&=\bigcup_{\{v\in {\cal V}: \phi(v)\in \Lambda_{0} + \omega\}}{\cal S}(v)\label{eqS}
\end{align}
where  ${\cal S}(v) \triangleq \left\{[u,w]: [\phi(u),\phi(w)]\in {\cal S}\big(\phi^{-1}(v)\big), u,w\in {\cal V}\right\}$ and  ${\cal S}(v)\bigcap{\cal S}(v')=\emptyset$, for all $v\neq v'$ with  $\phi(v),\phi(v')\in \Lambda_{0} + \omega$. Hence, as we can also see in Fig.~\ref{butterfly}, the set  ${\cal E}_{\pi^{*}}\setminus {\cal E}_{0} $ can be partitioned into smaller interference clusters  by   translating the  highlighted tilted rectangle over all black vertices. 
Apart from a few exceptions (due to the finite boundary of the network), the resulting  interference clusters
${\cal S}(v)$ will be isomorphic and exhibit the same structure as the one 
shown in Fig.~\ref{iasol}. 

In the following, we will describe the interference alignment solution that lies at the core of our achievability theorem, namely the precoding scheme that can achieve $d_{v}=1$ for all cells $v$ in a given  cluster. Then, we will see that the above solution can be readily  extended to the entire network through (\ref{eqS})  and therefore show that the required interference alignment conditions (\ref{iacond10}) and (\ref{iacond11}) are satisfied in ${\hat{\cal G}}_{\pi^{*}}({\cal V} \setminus {\cal V}_{0},{\cal E}_{\pi^{*}}\setminus {\cal E}_{0})$.

\begin{figure}[h]
                \centering
                \includegraphics[width=0.85\columnwidth]{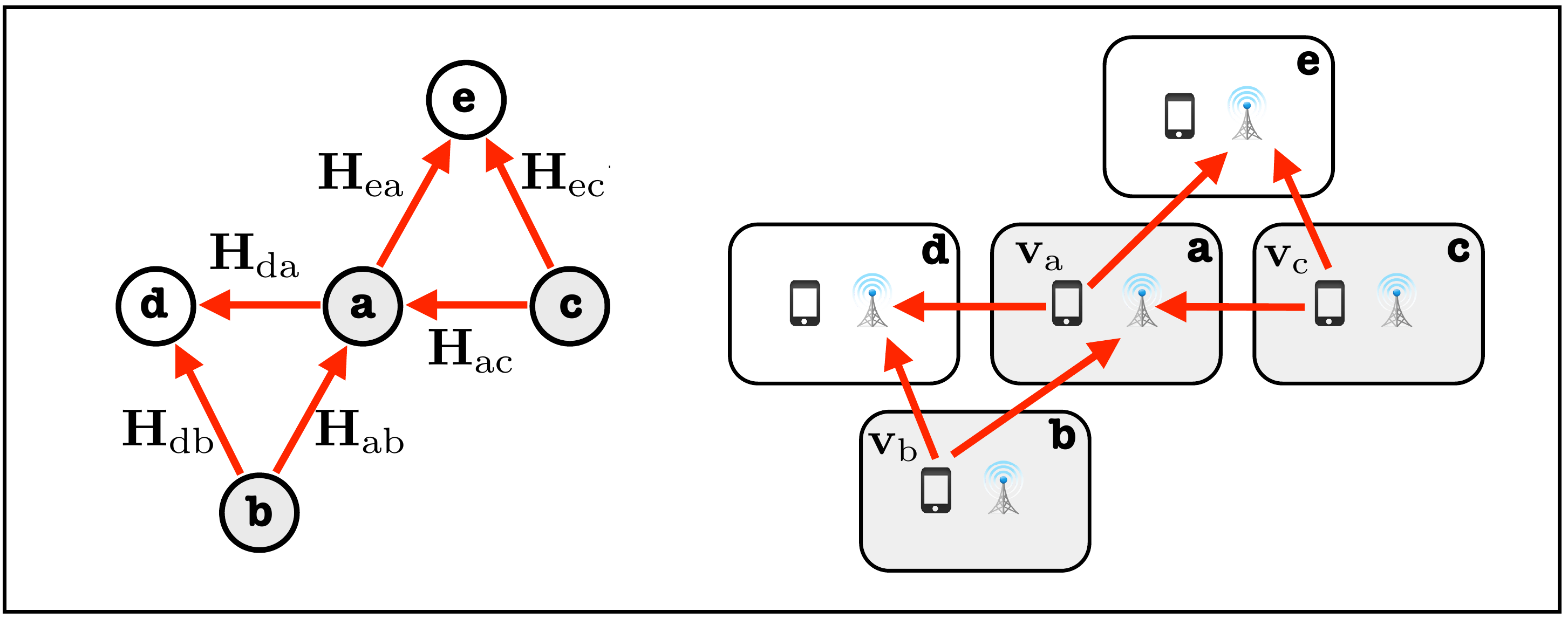}
                \caption{ The interference edges in the cluster ${\cal S}(a)$. The transmit beamforming vectors $\vv_{a}$, $\vv_{b}$ and $\vv_{c}$ (gray nodes) 
                are uniquely associated with the edges in ${\cal S}(a)$ and have to be designed such that interference is aligned at receivers $a$, $d$ and $e$.\vspace{-0in}}
                \label{iasol}

\end{figure}

In the above cluster, the goal is to design the $2$-dimensional beamforming vectors ${\bf v}_{ a}$, ${\bf v}_{ b}$, and ${\bf v}_{ c}$ such that all interference occupies a single dimension in every receiver. We will hence  require that 
$\mbox{span}(\Hm_{ab}\vv_{b}) = \mbox{span}({\bf H}_{ ac}{\bf v}_{ c})$ for receiver $a$,
$\mbox{span}(\Hm_{db}\vv_{b}) = \mbox{span}({\bf H}_{ da}{\bf v}_{ a})$ for receiver $d$, and
$\mbox{span}(\Hm_{ea}\vv_{a}) = \mbox{span}({\bf H}_{ ec}{\bf v}_{ c})$ for receiver $e$.
These  alignment conditions can be written as,
\begin{eqnarray}
{\bf v}_{ a} &\doteq& {\bf H}_{ da}^{-1}{\bf H}_{ db}{\bf v}_{ b}\label{Iacond1}\\
{\bf v}_{ b} &\doteq& {\bf H}_{ ab}^{-1}{\bf H}_{ ac}{\bf v}_{ c}\label{Iacond2} \\
{\bf v}_{ c} &\doteq& {\bf H}_{ ec}^{-1}{\bf H}_{ ea}{\bf v}_{ a},\label{Iacond3}
\end{eqnarray}
where ${\bf v}\doteq {\bf u}$ is a shorthand notation for ${\bf v} \in \mbox{span}({\bf u})$, and are satisfied as long as 
\begin{equation}
{\bf v}_{ a} \doteq {\bf H}_{ da}^{-1}{\bf H}_{ db}{\bf H}_{ ab}^{-1}{\bf H}_{ ac}{\bf H}_{ ec}^{-1}{\bf H}_{ ea}{\bf v}_{ a}.\label{eigsol}
\end{equation}
Therefore, if we choose ${\bf v}_{ a}$ to be an {\it eigenvector} of the above matrix and set 
\begin{eqnarray}
{\bf v}_{ b} &\doteq& {\bf H}_{ ab}^{-1}{\bf H}_{ ac}{\bf H}_{ ec}^{-1}{\bf H}_{ ea}{\bf v}_{ a} \mbox{  and}\\
{\bf v}_{ c} &\doteq& {\bf H}_{ ec}^{-1}{\bf H}_{ ea}{\bf v}_{ a},\label{eigsol3}
\end{eqnarray}
all interference observed at the receivers $a$, $d$ and $e$ will be  aligned in one dimension, and can hence be zero-forced by the corresponding receiver projections  $\uv_{a}$, $\uv_{d}$ and $\uv_{e}$. Assuming that the channel matrices are drawn from a continuous non-degenerate distribution, the projected useful signal coefficients  will be non-zero with probability one and therefore, the desired messages $s_{a}$, $s_{d}$ and $s_{e}$ can  be successfully decoded from the interference-free observations $\uv_{a}^{\rm H}\yv_{a}$, $\uv_{d}^{\rm H}\yv_{d}$ and $\uv_{e}^{\rm H}\yv_{e}$.

It is important to note that the transmit beamforming choices ${\bf v}_{ a}$, ${\bf v}_{ b}$, and ${\bf v}_{ c}$ as well as the receiver projections ${\bf u}_{ a}$, ${\bf u}_{ d}$, and ${\bf u}_{ e}$ are {\it uniquely associated} with the interference edges in the above cluster  and do not participate in  any other interference alignment conditions in the network. 
Therefore,  the  eigenvector solution  (\ref{eigsol})-(\ref{eigsol3}) can be applied locally for all interference clusters ${\cal S}(v)$ in the network in order to choose the appropriate beamforming vectors  $\vv_{v}$ and $\uv_{v}$ for all active cells $v\in {\cal V}\setminus {\cal V}_{0}$. Since the interference alignment conditions (\ref{iacond10}) and (\ref{iacond11}) can be satisfied across the entire network, the DoFs $$d_{v} = \begin{cases} 1, \;\;v \in {\cal V}\setminus {\cal V}_{0}\\ 0, \;\;v\in {\cal V}_{0}\end{cases}$$
are achievable in $\cal G(V,E)$ and hence 
\begin{equation}
d_{{\cal G},{\pi}^{*}} = \frac{1}{|{\cal V}|}\sum_{v\in \cal V}d_{v} = \frac{|{\cal V}\setminus {\cal V}_{0}|}{|{\cal V}|}=1-\frac{|{\cal V}_{0}|}{|{\cal V}|}.\label{eqDOFs}
\end{equation}

\begin{lemma}
The number of inactive cells  $|{\cal V}_{0}|$ can always be chosen in $\cal G(V,E)$ to satisfy $|{\cal V}_{0}|\leq|{\cal V}|/4$.
\end{lemma}
\begin{proof}
Recall that the total number of cells is  $|{\cal V}| = \left|\ZZ(\omega)\cap{\cal B}_{r}\right|$ 
and the number of  inactive cells is $|{\cal V}_{0}| = \left|\Lambda_{0}\cap{\cal B}_{r}\right|$,
where
${\cal B}_{r} \triangleq \{z\in \mathbb{C} :|{\rm Re}(z)|\leq r , |{\rm Im}(z)|\leq \frac{\sqrt{3}r}{2}\}$
is defined in Section \ref{sec:igraph} and $\Lambda_{0}\triangleq 2\cdot\ZZ(\omega)$ is chosen in (\ref{L0}). Let $R({r})$
be the ratio $\frac{|{\cal V}_{0}|}{|{\cal V}|}= \frac{\left|\Lambda_{0}\cap{\cal B}_{r}\right|}{\left|\mathbb{Z}(\omega)\cap{\cal B}_{r}\right|}$ as a function of the network size parameter $r$. For  large cellular networks  we can already see that $$\lim_{r\rightarrow\infty}R({r})= \frac{Vol\left(\mathbb{Z}(\omega)\right)}{Vol\left(2\cdot\mathbb{Z}(\omega)\right)}=\frac{1}{4}.$$ 
We want to show however that $R({r})\leq \frac{1}{4}$, for all  $r\geq 1$. Notice that the points in $\ZZ(\omega)\cap{\cal B}_{r}$ can be partitioned into $2r+1$ horizontal lines. When $r$ is odd, the set $\ZZ(\omega)\cap{\cal B}_{r}$ has $r+1$  lines with $2r$ points and $r$  lines with $2r+1$ points. Therefore the total number of points can be calculated as $|\ZZ(\omega)\cap{\cal B}_{r}|= (r+1)2r+r(2r+1)=4r^{2}+3r$. Now, all the points that correspond to the $r+1$ lines contain odd multiples of $\omega$ and hence cannot be in $\Lambda_{0}$. From the remaining $r$ lines, $\frac{r+1}{2}$ lines have $r$ points in $\Lambda_{0}$ and $\frac{r-1}{2}$ lines have $r+1$ points in $\Lambda_{0}$. The total number of points in $\Lambda_{0}\cap{\cal B}_{r}$ is therefore given by $|\Lambda_{0}\cap{\cal B}_{r}| =\frac{r+1}{2}r+\frac{r-1}{2}(r+1)=r^{2} +(r-1)/2$, and the corresponding ratio can be calculated~as 
$$R(r) = \frac{r^{2} +(r-1)/2}{4r^{2}+3r}\leq\frac{r+1/2}{4r+2}=\frac{1}{4}.$$
For the case where $r$ is even the same result can be obtained if we choose $\Lambda_{0}=2\cdot\ZZ(\omega)+\omega$. We omit the details here for brevity.  
\end{proof}

From the above lemma and the result obtained in (\ref{eqDOFs}), we conclude that the average (per cell) DoFs $d_{{\cal G},\pi^{*}} = 3/4$, are indeed achievable  in  $\cal G(V,E)$ under the network interference cancellation framework with decoding order $\pi^{*}$, and the proof of Theorem~\ref{thm1} is completed.

\subsection{Converse (Proof of Theorem 2)}\label{converse}
We begin by stating a useful lemma that will help us bound the total DoFs achievable in our setting.

\begin{lemma} The degrees of freedom achievable by linear schemes in an $M\times M$ flat-fading MIMO cellular network with directed interference graph ${\cal G}_{\pi}({\cal V},{\cal E}_{\pi})$, must satisfy
\begin{align}
&d_{v} \in \{0,\dots, M\},\;\forall v\in {\cal V}\label{IAfeasibility}\\
&d_{u}+d_{v} \leq M,\;\forall [u,v]\in {\cal E}_{\pi}\label{IAfeasibility2}\\
&2\sum_{v\in \cal V}(M - d_{v})d_{v} \geq \sum_{(u,v)\in \cal E_{\pi}} d_{u}d_{v} \label{IAfeasibility3}
\end{align}
%
%
for any decoding order $\pi$, assuming all channel gains are chosen from a continuous distribution.
\end{lemma}

\begin{proof}
Consider the decoding order $\pi$ and assume that a genie provides $W_{\pi(1)}$ to $W_{\pi(i-1)}$ to the receiver  $\pi(i)$, for all $i=1,...,|\cal V|$. Since, in the original problem, the decoding order is  $\pi$, and we have assumed  that each receiver can only share its own decoded message with its neighbors, then it is clear that the DoF region of the new genie-aided interference network is an outer-bound for the original one. This side information reduces the original interference graph  $\cal G(V,E)$ to the directed interference graph ${\cal G}_{\pi}({\cal V},{\cal E}_{\pi})$.  The lemma follows by using the linear IA  feasibility conditions developed in~\cite{Razaviyayn},~\cite{Bresler}. 
\end{proof}

The above inequalities are {\it necessary} conditions for the achievability of any degrees of freedom $\{d_{v},v\in\cal V\}$ in ${\cal G}_{\pi}({\cal V},{\cal E}_{\pi})$. We are going to use these conditions here to obtain an upper bound on the average degrees of freedom achievable in our network by considering the  optimization problem
\begin{align}
\hspace{-0.2in}{{ \rm Q}_{1}({{\cal G}_{\pi}})}:\;\;\;\;\; &\underset{\{d_{v},v\in\cal V\}}{\mbox{maximize}}\;\;\;  \frac{1}{|{\cal V}|} \sum_{v\in \cal V}d_{v} \nonumber\\
&\mbox{subject to:} \;\;(\ref{IAfeasibility}),(\ref{IAfeasibility2}),(\ref{IAfeasibility3}).
\nonumber
\end{align}

As we can  see, ${ \rm Q}_{1}({{\cal G_{\pi}}})$  is a very difficult problem to solve in its current form: Not only it involves  non-linear integer constrains but also it is defined over many variables (since we want to consider large networks). 
In the following, we are going to transform (and relax) this problem to a tractable linear program (LP) with significantly fewer variables that do not scale with the size of the network. Our approach can be summarized in the following steps:
\begin{itemize}
\item First, we decompose the network into a set of three-user interference sub-networks, such that the three cells in each sub-network are geographically nearby and therefore interfere with each other. In the interference graph topology each of such subnetwork forms a triangle. 

\item Then, we identify the set of all possible total degrees of freedom that each isolated triangle can achieve. This set induces a finite number of possible distinct options, that we will refer to as triangle-DoF configurations. 

\item  Next, we rewrite the entire  objective function and the constraints as  linear functions of the relative frequencies in which  each triangle-DoF configuration occurs in the network. This reformulation has two advantages: First, since the number of triangle DoF configurations is limited,  the number of variables in the resulting optimization problem will not scale with the size of the network. Second, the relative frequency of the each triangle-DoF configuration is a positive rational number between zero and one, and therefore relaxing this number to a real number between zero and one does not lead to a loose upper-bound.  


\item Finally, after we obtain the corresponding linear program from the above reformulation, we use standard tools from duality theory to upper bound its optimal value and therefore conclude that 
$d_{{\cal G},\pi}\leq{\rm opt}({ \rm Q}_{1}({{\cal G_{\pi}}}))\leq 3/4 + {\cal O}\left( \scriptstyle{1}/{{\sqrt{|{\cal V}|}}}\right)$, for all decoding orders $\pi$.
\end{itemize}


\subsubsection{Decomposing the  Network into Triangles}

\begin{figure}[ht]

                \centering
                \includegraphics[width=.55\columnwidth]{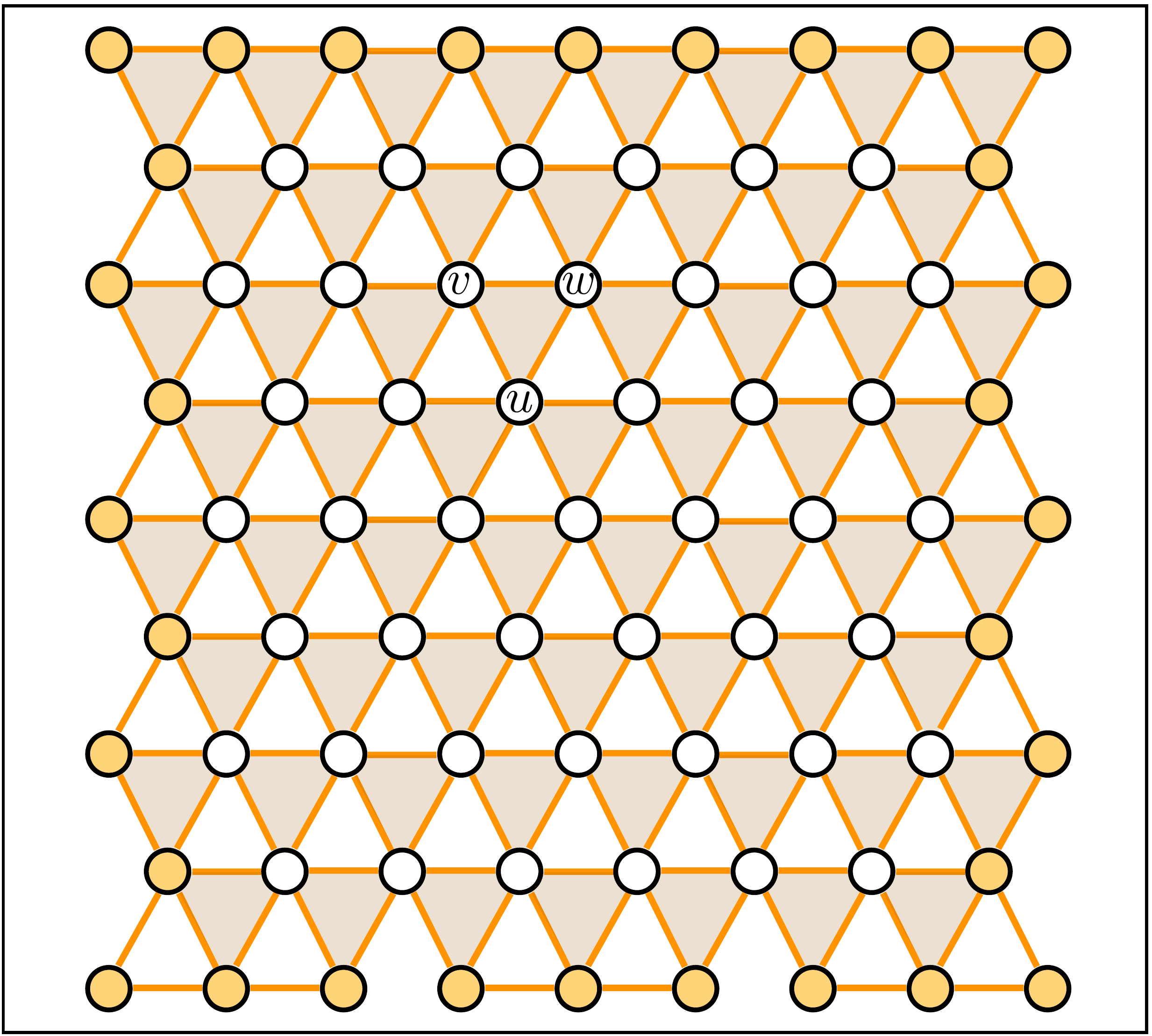}
                \caption{The set of triangles $[a,b,c]\in \cal T$ for  $\cal G(V,E)$. All the circle nodes belong to $\cal V_{\rm in}$ and participate in exactly three triangles ($n_{v}=3$). The set $\cal V_{\rm ex}$ contains the  colored nodes on the 
                boundary for which $n_{v}<3$.\vspace{-0.15in}}
                \label{Triangles}

\end{figure}
As a first step we will decompose the entire cellular network into a set of three-user interference sub-networks $\cal T$ that form adjacent triangles in the corresponding interference graph $\cal G(V,E)$ as shown in Fig.~\ref{Triangles}.
In order to formally describe the set of triangles $\cal T$ in $\cal G(V,E)$, we consider the set of ordered Eisenstein integer triplets $${ \cal P} = \{[z,z+\omega,z+\omega+1]: z\in\mathbb{Z}(\omega) \}.$$
Recall from Section \ref{sec:igraph} that  the points $z$, $z+\omega$ and $z+\omega+1$ form the line segments $\Delta(z)\subseteq {\cal D}$ and that the corresponding graph vertices $\phi^{-1}(z)$, $\phi^{-1}(z+\omega)$ and $\phi^{-1}(z+\omega+1)$  form a connected  triangle in $\cal G(V,E)$. The set $\cal T$ can hence be defined as
\begin{equation}
{\cal T} \triangleq \{[a,b,c]: [\phi(a),\phi(b),\phi(c)]\in {\cal P}, a,b,c\in {\cal V}\}.
\label{triangledef}
\end{equation}

The above definition is illustrated in Fig.~\ref{Triangles} in which shaded triangles connect the corresponding vertex triplets $[a,b,c]\in \cal T$.
Notice that apart from the vertices on the external boundary of the graph, all other nodes participate in exactly three triangles in $\cal T$. This observation will be particularly useful in rewriting the sum in the objective function of ${ \rm Q}_{1}({{\cal G_{\pi}}})$ as a sum over $\cal T$ instead of $\cal V$.

Let \vspace{-0.1in}
\begin{equation}
n_{v}\triangleq \sum_{[a,b,c]\in {\cal T}}\mathlarger{\mathbbm{1}}\Big\{v\in\{a,b,c\}\Big\}
\label{nv}
\end{equation}
denote the number of triangles $[a,b,c]\in {\cal T}$ that include a given vertex $v\in \cal V$.
As we have seen, $n_{v}$ can only take values in $\{0, 1, 2, 3\}$ for any $v\in \cal V$. More specifically $n_{v}= 3$ for all  internal vertices in $\cal G(V,E)$, while  $n_{v}< 3$ only for  external vertices that lie on the outside boundary of the graph.
We define the set of internal and external vertices as 
\begin{eqnarray}
{\cal V}_{\rm in} &=& \{v\in {\cal V} : n_{v}=3\}, \mbox{ and} \\
{\cal V}_{\rm ex} &=& \{v\in {\cal V} : n_{v}<3\},
\end{eqnarray}
and in Fig.~\ref{Triangles} we show the above distinction by coloring all graph vertices $v$ that belong to ${\cal V}_{\rm ex}\subseteq\cal V$.  

\subsubsection{Reformulating  ${\rm Q}_{1}({{\cal G}_{\pi}})$ as a Linear Program} Notice that the objective and the constraint functions in ${\rm Q}_{1}({{\cal G}_{\pi}})$ are given as a sum over the vertices $v\in \cal V$. As a first step towards reformulating  ${\rm Q}_{1}({{\cal G}_{\pi}})$ as an LP, we will use the following two lemmas and rewrite the functions of ${\rm Q}_{1}({{\cal G}_{\pi}})$ in terms of sums over the corresponding triangles $[a,b,c]\in \cal T$.
\vspace{-0.1in}
\begin{lemma}
The degrees of freedom $\{d_{v},v\in \cal V\}$, satisfy
 \begin{align}
  \frac{1}{|{\cal V}|} \sum_{v\in \cal V}d_{v} \leq  \frac{1}{3|{\cal V}|}\sum_{[a,b,c]\in \cal T}
s(d_{a},d_{b},d_{c})  + \frac{M|{\cal{V}_{\rm ex}}|}{|\cal V|},
  \end{align}
  where 
    \begin{align}
  s(d_{a},d_{b},d_{c}) \triangleq \;&d_{a}+d_{b}+d_{c}.\;\;
  \end{align}
  \label{lem:objective}
\end{lemma}
\begin{proof}
See Appendix \ref{proof:objective}
\end{proof}

\begin{lemma}
Any degrees of freedom $\{d_{v},v\in \cal V\}$ that satisfy (\ref{IAfeasibility3}) also satisfy: 
 \begin{align}
  &\sum_{[a,b,c]\in \cal T}   g(d_{a},d_{b},d_{c})
  \leq \frac{3M^{2}}{2}|{\cal{V}_{\rm ex}}|,
  \end{align}
  where  \vspace{-0.1in}
    \begin{align}
  g(d_{a},d_{b},d_{c}) \triangleq \;&(d_{a}+d_{b})^{2}+(d_{a}+d_{c})^{2}+(d_{b}+d_{c})^{2} \nonumber\\ &+{d_{a}d_{b}+d_{a}d_{c}+d_{b}d_{c}} -2M(d_{a}+d_{b}+d_{c}).
  \end{align}
  \label{lem:constraint}
\vspace{-0.2in}
\end{lemma}
\begin{proof}
See Appendix \ref{proof:constraint}
\end{proof}

Replacing the objective and constraint functions of ${ \rm Q}_{1}({{\cal G}_{\pi}})$ with the corresponding functions given above, we arrive at
\begin{align}
\hspace{-0.2in}{{ \rm Q}_{2}({{\cal G}_{\pi}})}:\;\;\;\;\; &\underset{\{d_{v},v\in\cal V\}}{\mbox{maximize}}\;\;\;   \frac{1}{3|{\cal V}|}\sum_{[a,b,c]\in \cal T}
s(d_{a},d_{b},d_{c}) +\frac{M|{\cal{V}_{\rm ex}}|}{|\cal V|}\label{c0} \\
&\mbox{subject to:} \;\;
d_{v} \in \{0,\dots, M\},\;\forall v\in {\cal V}\label{c1}\\
&\;\;\;\;\;\;\;\;\;\;\;\;\;\;\;\;\;\,d_{u}+d_{v} \leq M,\;\forall [u,v]\in {\cal E}_{\pi}\label{c2}\\ 
&\;\;\;\;\;\;\;\;\;\;\;\;\;\;\;\;\;\,\sum_{[a,b,c]\in \cal T}g(d_{a},d_{b},d_{c})\leq \frac{3M^{2}}{2}|{\cal{V}_{\rm ex}}| \,,\label{c3} 
\end{align}
whose optimal value satisfies   ${\rm opt}({ \rm Q}_{2}({{\cal G}_{\pi}}))\geq{\rm opt}({ \rm Q}_{1}({{\cal G}_{\pi}}))$, i.e., the solution of ${ \rm Q}_{2}$ provides an outer bound for the achievable DoFs in the cellular network.

A key observation is that  $s(d_{a},d_{b},d_{c})$ and $g(d_{a},d_{b},d_{c})$ can only take specific values for each triangle due to the constraints (\ref{c1}) and (\ref{c2}) and are invariant under permutations of their arguments. As the next step in our proof, we will 
define the set $\cal D$ of all distinct triangle-DoF configurations that will subsequently limit the possible values that $s(d_{a},d_{b},d_{c})$ and $g(d_{a},d_{b},d_{c})$ can take. When $M=2$ it is easy to see that all possible DoF configurations $(d_{a},d_{b},d_{c})$ will belong to the set  $\{[0,0,0], [0,0,1], [0,0,2], [0,1,1], [1,1,1]\}.$ For general $M$, the set $\cal D$ is given by

\begin{equation}\label{TD}
{\cal D} \triangleq\left\{[i,j,k]: 
\begin{aligned} &i\leq j\leq k \in \{0,\dots,M\}\\ &i+j\leq M \\
&j+k\leq M \\
&k+i\leq M  \end{aligned}\right\}.
\end{equation}

Now we can define the relative frequencies

\begin{equation}
x_{(i,j,k)} \triangleq  \frac{1}{|\cal T|}\sum_{[a,b,c]\in \cal T} \mathlarger{\mathbbm{1}} \Big\{
(d_{a},d_{b},d_{c})=(i,j,k)\Big\},
\end{equation}
 to be the fraction of  triangles $[a,b,c]\in \cal T$ with the specific DoF configuration $[i,j,k]\in \cal D$, and write the corresponding sums in (\ref{c0}) and (\ref{c3}) as 
  \begin{align}
&\sum_{[a,b,c]\in \cal T} s{(d_{a},d_{b},d_{c})} = {|\cal T|} \sum_{[i,j,k] \in \cal D} s{(i,j,k)}\cdot x_{(i,j,k)}\;\\
&\mbox{and}\nonumber\\
&\sum_{[a,b,c]\in \cal T} g{(d_{a},d_{b},d_{c})} =  {|\cal T|} \sum_{[i,j,k] \in \cal D} g{(i,j,k)}\cdot x_{(i,j,k)}\,.
\end{align}

Finally, relaxing the integrality constraints in (\ref{c1}) and letting  $x_{(i,j,k)}\in \RR$,  we can reduce the optimization problem ${ \rm Q}_{2}({{\cal G}_{\pi}})$ to a linear program (LP) over the variables $x_{(i,j,k)}, \, [i,j,k] \in \cal D$
given by

\vspace{-0.2in}
\begin{align}
{{ \rm LP}({{\cal G}_{\pi}})}:\;\; \underset{\left\{ \substack{x_{(i,j,k)}\in \mathbb{R},\\ [i,j,k]\in \cal D} \right\}}{\mbox{maximize}}\;\;\;  & \frac{|{\cal T}|}{3|{\cal V}|} \sum_{[i,j,k] \in \cal \cal D}s{(i,j,k)}\cdot x_{(i,j,k)} + \frac{M|{\cal V}_{\rm ex}|}{|{\cal V}|} \\
\mbox{subject to:} \;\;& \sum_{[i,j,k] \in \cal\cal D} g{(i,j,k)}\cdot x_{(i,j,k)} \leq \frac{3M^{2}|{\cal V}_{\rm ex}|}{2|{\cal T}|}
\\
&\sum_{[i,j,k] \in \cal D} x_{(i,j,k)}  = 1 \\
& \;\;\,x_{(i,j,k)}  \geq  0, \; \forall {[i,j,k] \in \cal \cal D}   
\end{align}
whose objective is to identify the relative frequencies $x_{(i,j,k)}$ of the triangle-DoF configurations that, under the (relaxed) feasibility constraints, maximize the (upper-bounded) degrees of freedom in $\cal G$.

\subsubsection{The LP for the $M=2$ case} As we have seen, when $M=2$,  the feasible triangle-DoF configurations are given by $${\cal D} = \{[0,0,0], [0,0,1], [0,0,2], [0,1,1], [1,1,1]\}.$$ Therefore, if we define the $5$-dimensional vectors  
\begin{align}
\sv \triangleq& \begin{bmatrix}s(0,0,0),\, s(0,0,1),\, s(0,0,2),\, s(0,1,1),\,s(1,1,1)\end{bmatrix}^{\rm T} = \begin{bmatrix}0,1,2,2,3\end{bmatrix}^{\rm T}\\
\gv \triangleq& \begin{bmatrix}g(0,0,0),\, g(0,0,1),\, g(0,0,2),\, g(0,1,1),\,g(1,1,1)\end{bmatrix}^{\rm T} = \begin{bmatrix}0,-2,0,-1,3\end{bmatrix}^{\rm T}
\end{align}
and  let 
$\xv \triangleq \begin{bmatrix}x_{(0,0,0)},\, x_{(0,0,1)},\, x_{(0,0,2)},\, x_{(0,1,1)},\,x_{(1,1,1)}\end{bmatrix}^{\rm T}$, we arrive at the linear program :

\begin{align}
{\rm LP}_{2}({{\cal G}_{\pi}}):\;\;\;\;\; \underset{\left\{ \substack{\xv \in \mathbb{R}^{5}} \right\}}{\mbox{maximize}}\;\;\;\;\;  &\frac{|{\cal T}|}{3|{\cal V}|}\;\sv^{\rm T}\xv + 2\frac{|{\cal V}_{\rm ex}|}{|{\cal V}|}\label{lpg1} \\
\;\mbox{subject to:} \;\;\;\;\;& \gv^{\rm T}\xv \leq 6\frac{|{\cal V}_{\rm ex}|}{|{\cal T}|},\;\,{\bf 1}^{\rm T}\xv  = 1,\;\,\xv  \geq  0. \label{lpg2} 
\end{align}

To obtain some intuition for the above optimization problem, notice that in Fig.~\ref{Triangles}, the  triangles $\cal T$ in $\cal G(V,E)$ are almost as many as the vertices $\cal V$ and therefore we can argue that $\frac{|{\cal T}|}{3|{\cal V}|}\rightarrow \frac{1}{3}$ as $|V|\rightarrow\infty$. On the other hand, the number of external vertices $|\cal V_{\rm ex}|$ is much smaller that    
$|\cal V|$ and we have that  $\displaystyle\lim_{|{\cal V|}\rightarrow \infty}|{\cal V_{\rm ex}}|/|{\cal V}|=0$. Consequently, one should expect that in large cellular systems, the average achievable DoFs are upper bounded by the solution of the linear program:

\begin{equation}
\underset{\left\{ \substack{\xv \in \mathbb{R}^{5}} \right\}}{\mbox{maximize}}\;\;
\frac{1}{3}\;\sv^{\rm T}\xv,\; \mbox{subject to:} \;\;\gv^{\rm T}\xv \leq 0,\;\,{\bf 1}^{\rm T}\xv  = 1,\;\,\xv  \geq  0.  
\end{equation}

Quite remarkably, the solution of the above LP is given by the relative frequency vector $$\xv^{*} = \big[x_{(0,0,0)},\, x_{(0,0,1)},\, x_{(0,0,2)},\, x_{(0,1,1)},\,x_{(1,1,1)}\big]^{\rm T} = \big[0,\,0,\,0,\,3/4,\,1/4\big],$$ and its optimal value is $\frac{1}{3}\;\sv^{\rm T}\xv^{*}= {3}/{4}$. It is  interesting to note that the optimal relative frequencies $x_{(0,1,1)}=3/4$ and $x_{(1,1,1)}=1/4$ given in the above solution are exactly the same as the ones used in our achievability scheme in Fig.~\ref{butterfly}.

\subsubsection{Solving the Linear Program} To conclude the proof of Theorem~\ref{thm2}, in this section
 we will focus on the $M=2$ case and use standard LP duality tools to formally show  that the solution of ${\rm LP}_{2}({{\cal G}_{\pi}})$ is upper bounded by $3/4 + {\cal O}\left( \scriptstyle{1}/{{\sqrt{|{\cal V}|}}}\right)$. For additional results and the extension of the corresponding bounds to the general $M\times M$ case we refer the reader to Appendix~\ref{generalM}.

%
%
%
%
\begin{lemma}\label{lem:LP-abc}
Let ${\rm opt}(\sv,\gv,\alpha,\beta,\gamma)$ denote the optimal value of the linear program given by
\begin{align}
\underset{\left\{ \substack{\xv \in \mathbb{R}^{n}} \right\}}{\mbox{maximize}}\;\;\;\;\;  &\alpha\cdot \sv^{\rm T}\xv + \beta \\
\;\mbox{subject to:} \;\;\;\;\;& \gv^{\rm T}\xv \leq \gamma,\;\,{\bf 1}^{\rm T}\xv  = 1,\;\,\xv  \geq  0.  
\end{align}
where $\alpha,\beta,\gamma\geq0$ and $\sv,\gv\in \RR^{n}$. Then, for any $\lambda\geq0$, we have that  $${\rm opt}(\sv,\gv,\alpha,\beta,\gamma)\leq \alpha\cdot \max_{i}\left\{s_{i}-{\lambda}\cdot g_{i}\right\} + \lambda\cdot\alpha\gamma+\beta.$$\end{lemma}
\begin{proof}
See Appendix \ref{proof:lem:LP-abc}.  \end{proof}

Using the result of the above lemma, we can show that the solution of the linear program ${\rm LP}_{2}({{\cal G}_{\pi}})$ is upper bounded by
\begin{equation} \label{boundLP2}
{\rm opt}({\rm LP}_{2}({{\cal G}_{\pi}}))\leq \frac{|{\cal T}|}{3|{\cal V}|} \cdot \max_{i}\left\{s_{i}-\lambda\cdot{g_{i}} \right\} + 2(\lambda+1)\frac{|{\cal V}_{\rm ex}|}{|{\cal V}|},\vspace{0.1in}
\end{equation}
for any $\lambda\geq 0$, by identifying $\alpha = \frac{|{\cal T}|}{3|{\cal V}|}$, $\beta=2\frac{|{\cal V}_{\rm ex}|}{|{\cal V}|}$  and $\gamma= 6\frac{|{\cal V}_{\rm ex}|}{|{\cal T}|}$.
In order to obtain a tight bound for ${\rm opt}({\rm LP}_{2}({{\cal G}_{\pi}}))$ we are going to choose 
\begin{align}\lambda^{*} =\arg\min_{\lambda\geq 0}\left\{\frac{1}{3}\cdot \max_{i}\left\{s_{i}-\lambda\cdot{g_{i}} \right\} \right\}
\end{align}

\begin{figure}[h]

                \centering
                \includegraphics[width=.55\columnwidth]{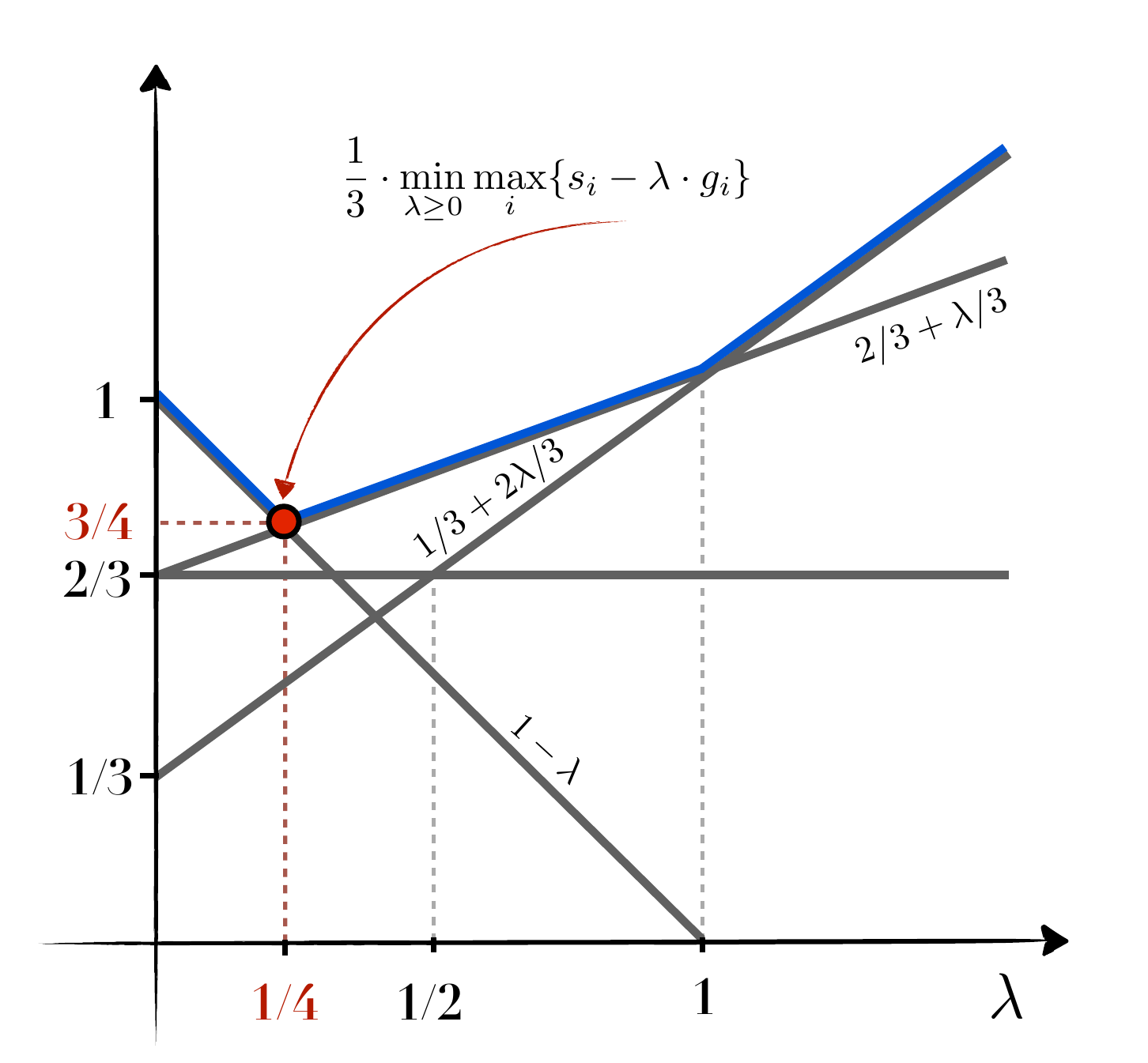}
                \caption{The minimum value of the piecewise linear function $\frac{1}{3}\cdot \max_{i}\left\{s_{i}-\lambda\cdot{g_{i}} \right\}.$}
                \label{minimax}

\end{figure}

The function $\frac{1}{3}\cdot \max_{i}\left\{s_{i}-\lambda\cdot{g_{i}} \right\} = \max\{1/3+2\lambda/3,\;2/3,\;2/3+\lambda/3\;,1-\lambda\}$ is a convex piecewise linear function in $\lambda$ (shown in Fig.~\ref{minimax}) and one can verify that $\lambda^{*}=1/4$ and $\frac{1}{3}\max_{i}\left\{s_{i}-\lambda^{*}\cdot{g_{i}} \right\} =3/4$. Substituting back in (\ref{boundLP2})
we obtain 
\begin{align}
{\rm opt}({\rm LP}_{2}({{\cal G}_{\pi}}))\leq& \frac{|{\cal T}|}{3|{\cal V}|} \cdot \max_{i}\left\{s_{i}-\lambda^{*}\cdot{g_{i}} \right\} + 2(\lambda^{*}+1)\frac{|{\cal V}_{\rm ex}|}{|{\cal V}|}\\
=& \frac{3}{4}\cdot \frac{|{\cal T}|}{|{\cal V}|}  + \frac{5}{2}\cdot\frac{|{\cal V}_{\rm ex}|}{|{\cal V}|}\label{gut}.\vspace{0.1in}
\end{align}

\begin{lemma}   \label{construction1}
By construction, the interference graph $\cal G(V,E)$ satisfies
  \begin{align}&|{\cal T}|\leq |\cal V|, \;\;\;\mbox{and}\\
  &|{\cal V_{\rm ex}}| = {\cal O}\left(\sqrt{|\cal V|}\right).
   \end{align}
\end{lemma}
\begin{proof}
See Appendix \ref{proof:construction}.  \end{proof}

Therefore, from (\ref{gut}) and Lemma~\ref{construction1} we obtain 
\begin{equation}
{\rm opt}({\rm LP}_{2}({{\cal G}_{\pi}}))\leq \frac{3}{4}  + {\cal O}\left({1}/{\sqrt{|{\cal V}|}}\right),
\end{equation}
which concludes the proof of Theorem~\ref{thm2}.

\section{Networks with Asymmetric Antenna Configurations}\label{sec:asym}
\begin{thm} \label{thm4}
For a $2\times3$ cellular  system $\cal G(V,E)$, there exist a 
one-shot linear beamforming scheme that  is able to achieve the average (per cell) DoFs,
$d_{{\cal G},{\pi^{*}}} = 1$,
under the network interference cancellation framework with decoding order $\pi^{*}$.
\hfill \QED
\end{thm}

\begin{thm} \label{thm5}
For a $2\times4$ cellular system $\cal G(V,E)$, there exist a 
one-shot linear beamforming scheme that is able to achieve the average (per cell) DoFs,
$d_{{\cal G},{\pi^{*}}} = {7}/{6}$,
under the network interference cancellation framework with decoding order $\pi^{*}$.
\hfill \QED
\end{thm}


\subsection{Achievability for $M=2$, $N=3$. (Proof of Theorem~\ref{thm4})}

Here, we  will focus on  the case where mobile terminal transmitters and  base-station 
receivers are equipped with $M=2$ and $N=3$ antennas  and describe the linear beamforming scheme  that is able to achieve one DoF per cell for the entire network.

Consider the directed interference graph ${\cal G}_{\pi^{*}}({\cal V}, {\cal E}_{\pi^{*}})$ shown in Fig.~\ref{2x3} and assume that all  user terminals $v\in \cal V$  are simultaneously transmitting their signals  $\xv_{v}= \vv_{v}s_{v}$ to their corresponding receivers. Following  the ``left-to-right, top-down'' decoding order $\pi^{*}$ introduced in Section \ref{sec:NICE}, 
each base-station  will attempt to decode its own desired message from the  three-dimensional receiver observation
\begin{equation}
\yv_{u} = \Hm_{uu}\vv_{u}s_{u} + \sum_{v : [v,u]\in {\cal E}_{\pi^{*}}} \Hm_{uv}\vv_{v}s_{v} +\zv_{u},
\end{equation}
by projecting  along  $\uv_{u}\in \CC^{3\times 1}$.
As before, the goal is to design $\vv_{u}\in \CC^{2\times 1}$ and $\uv_{u}\in \CC^{3\times 1}$ for all $u\in \cal V$ such that all interference is zero forced and the corresponding messages can be decoded from the projected observations $\uv_{u}^{\rm H}\yv_{u} = \uv_{u}^{\rm H}\Hm_{uu}\vv_{u}s_{u} + \hat z_{u}$. 

In order to show achievability, we will first focus on the cells $a$, $b$, $c$ and $d$ shown in Fig.~\ref{2x3} and then describe how the corresponding solution can be extended across the entire network.

The receiver $a$  first projects its observation onto the two dimensional subspace orthogonal to $\Hm_{ab}\vv_{b}$, effectively zero-forcing interference from transmitter $b$. The corresponding two-dimensional projected signal at receiver $a$ is given by 
\begin{align}
\Pm^{\perp}_{ab}\yv_{a} = \Pm^{\perp}_{ab}\Hm_{aa}\vv_{a}s_{a} +  \underbrace{\Pm^{\perp}_{ab}\Hm_{ac}\vv_{c}s_{c} + \Pm^{\perp}_{ab}\Hm_{ad}\vv_{d}s_{d}}_{\mbox{interference}}  +\hat \zv_{a},
\end{align}
where $\Pm^{\perp}_{ab}$ is a $2\times 3$ matrix that satisfies $\Pm^{\perp}_{ab}\Hm_{ab}\vv_{b} = {\bf 0}$. Further,  transmitter $d$ can  choose its beamforming vector as a function of $\vv_{c}$, 
\begin{align}
\Pm^{\perp}_{ab}\Hm_{ad}\vv_{d} &\doteq \Pm^{\perp}_{ab}\Hm_{ac}\vv_{c}\; \Leftrightarrow\\ 
\vv_{d}&\doteq (\Pm^{\perp}_{ab}\Hm_{ad})^{-1} \Pm^{\perp}_{ab}\Hm_{ac}\vv_{c}
\end{align}
such that all interference in the above observation is aligned in one dimension, and can hence be zero-forced by projecting along the two-dimensional vector $\tilde \uv_{a}$. Overall, the three-dimensional receiver projection is given by \begin{equation}
\uv_{a} = (\Pm^{\perp}_{ab})^{\rm H}\tilde \uv_{a}.
\end{equation}

The two key steps in the above scheme are depicted in Fig.~\ref{2x3}. Receiver $a$ projects its observation on a two dimensional subspace to zero-force (ZF) interference from cell $b$, and transmitter $d$ chooses $\vv_{d}$ in order to align the remaining interference (IA) with the cell $c$.

\begin{figure}[ht]

                \centering
                \includegraphics[width=.8\columnwidth]{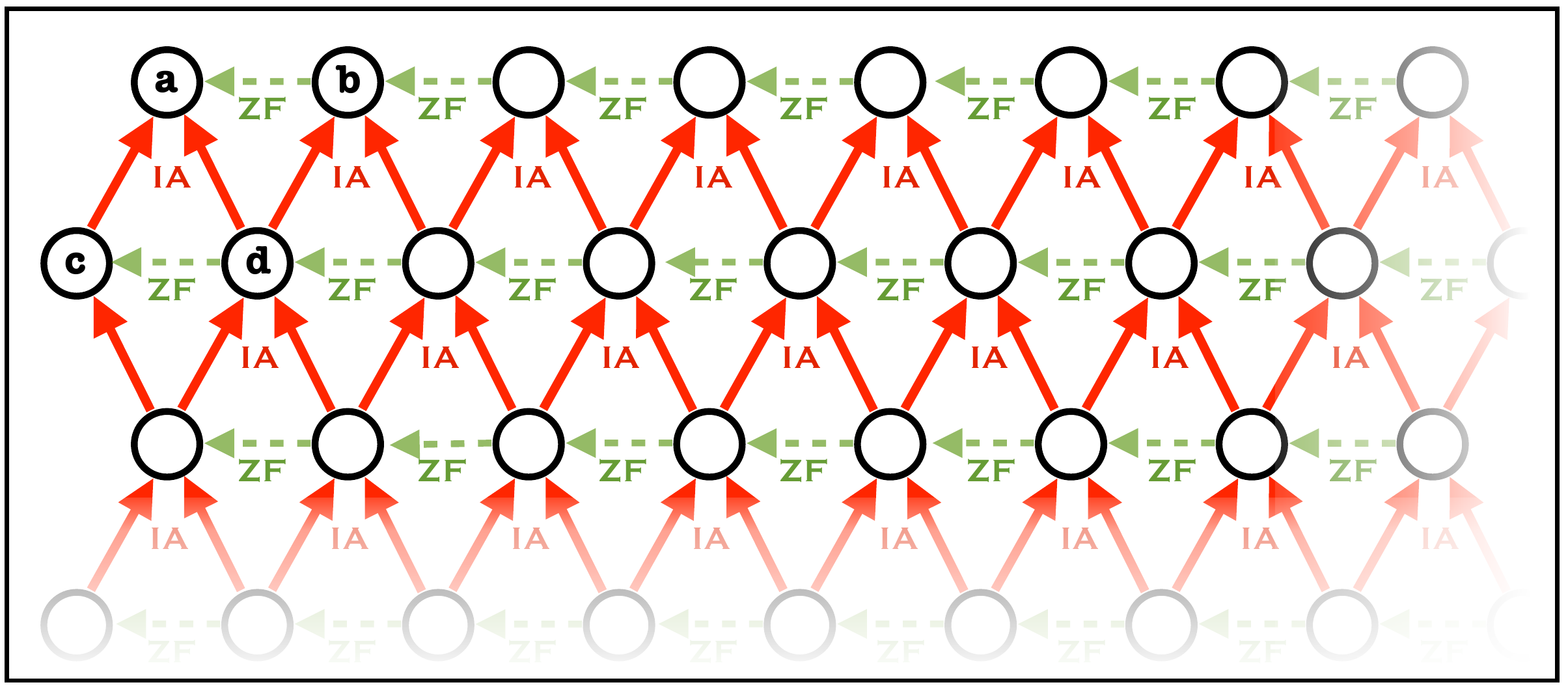}
                \caption{Interference alignment solution when $M=2$, and $N=3$.}
                \label{2x3}

\end{figure}

Exactly the same procedure can be followed to yield an interference alignment solution across the entire network. Starting from the boundary, the receivers $u\in \cal V$ can zero-force interference coming from the cell $v$ with $\phi(v)=\phi(u)+1$ (right neighbor) by projecting onto $\Pm^{\perp}_{uv}$, and the transmitters $v\in\cal V$ can  align their transmitted signals with the transmitter $u$ with $\phi(u)=\phi(v)-1$ (left neighbor) by choosing $\vv_{v}$ as a function of $\vv_{u}$.

\subsection{Achievability for $M=2$, $N=4$. (Proof of Theorem~\ref{thm5})}

In order to describe the interference alignment scheme for this case, we will first partition the directed interference graph ${\cal G}_{\pi^{*}}(\cal V,{\cal E}_{\pi^{*}})$ into {\it horizontal stripes}  as shown in Fig.~\ref{stripes}. Each stripe ${\cal S}_{k}$ contains three consecutive horizontal lines of nodes that we will refer to  as top, middle and bottom according to their relative position within each stripe.

\begin{figure}[ht]

                \centering
                \includegraphics[width=.7\columnwidth]{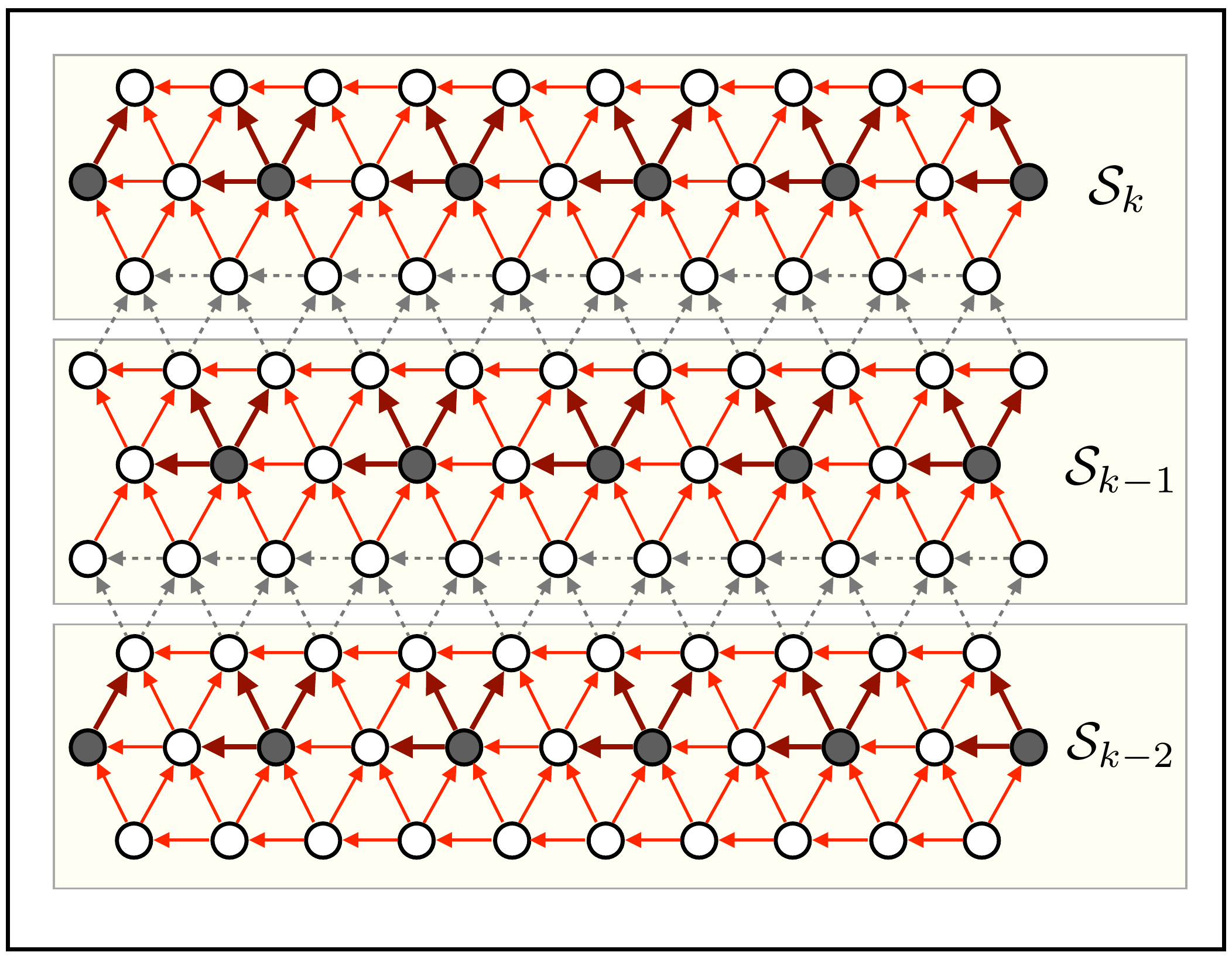}
                \caption{The directed interference graph $\cal G(V,E)$ divided in horizontal stripes.}
                \label{stripes}

\end{figure}

In the following we will describe a transmission scheme in which (at least) half of the nodes located in the middle line of each stripe (black nodes) are able to achieve $d_{v}=2$ while all the remaining nodes in the network (white nodes) are able to achieve $d_{v}=1$.   Let ${\cal V}_{1}$ and  ${\cal V}_{2}$ denote the corresponding subsets of vertices $v \in \cal V$ for which $d_{v}=1$ and $d_{v}=2$ respectively. Since 
 ${\cal V} = {\cal V}_{1} \cup {\cal V}_{2}$, the average (per cell) DoFs that will be achievable in $\cal G(V,E)$ with the above configuration can be written as
\begin{equation}
d_{{\cal G},{\pi}^{*}} = \frac{1}{|{\cal V}|}\sum_{v\in \cal V}d_{v} = \frac{|{\cal V}_{1}|+2|{\cal V}_{2}|}{|{\cal V}|}=1+\frac{|{\cal V}_{2}|}{|{\cal V}|},\label{eqDOFs2x4}
\end{equation}
and since 
${\cal V}_{2}$ can be chosen so that $|{\cal V}_{2}|\geq\frac{1}{6}|{\cal V}|$,
\footnote{Notice that the set ${\cal V}_{2}$  contains at least half of the nodes in every third line of the graph in Fig.~\ref{stripes}. This construction can be generalized for any interference graph $\cal G(V,E)$ so that $|{\cal V}_{2}|/|{\cal V}|\geq 1/6$.} 
we  obtain that $d_{{\cal G},{\pi}^{*}}\geq{7}/{6}$.

First, notice that all the receivers associated with nodes in the bottom line of each stripe ${\cal S}_{k}$ observe two interfering links coming from the transmitters in the top line of ${\cal S}_{k-1}$ and one interfering link coming from their adjacent (right) neighbor in ${\cal S}_{k}$, shown in Fig.~\ref{stripes} with dashed gray arrows. According to our DoF configuration, these interfering links will occupy  three out of the total four dimensions ($N=4$) in each receivers' subspace and can therefore be zero-forced in order to achieve $d_{v}=1$.  This important observation allows us to decouple the beamforming choices between different stripes and propose an interference alignment solution for a single stripe  that can be replicated across the entire network. 

Fig.~\ref{iasol2x4} shows the single stripe of the network that we will focus on for the rest of this section. The nodes are now shown with different shapes (diamonds, squares and circles) to illustrate their distinct roles in the achievability scheme that follows. Further, notice that the interfering edges in the bottom line of this stripe are not shown here since all interference in the corresponding nodes has already been zero forced (to decouple the consecutive stripes of the network). 

\begin{figure}[ht]

                \centering
                \includegraphics[width=.65\columnwidth]{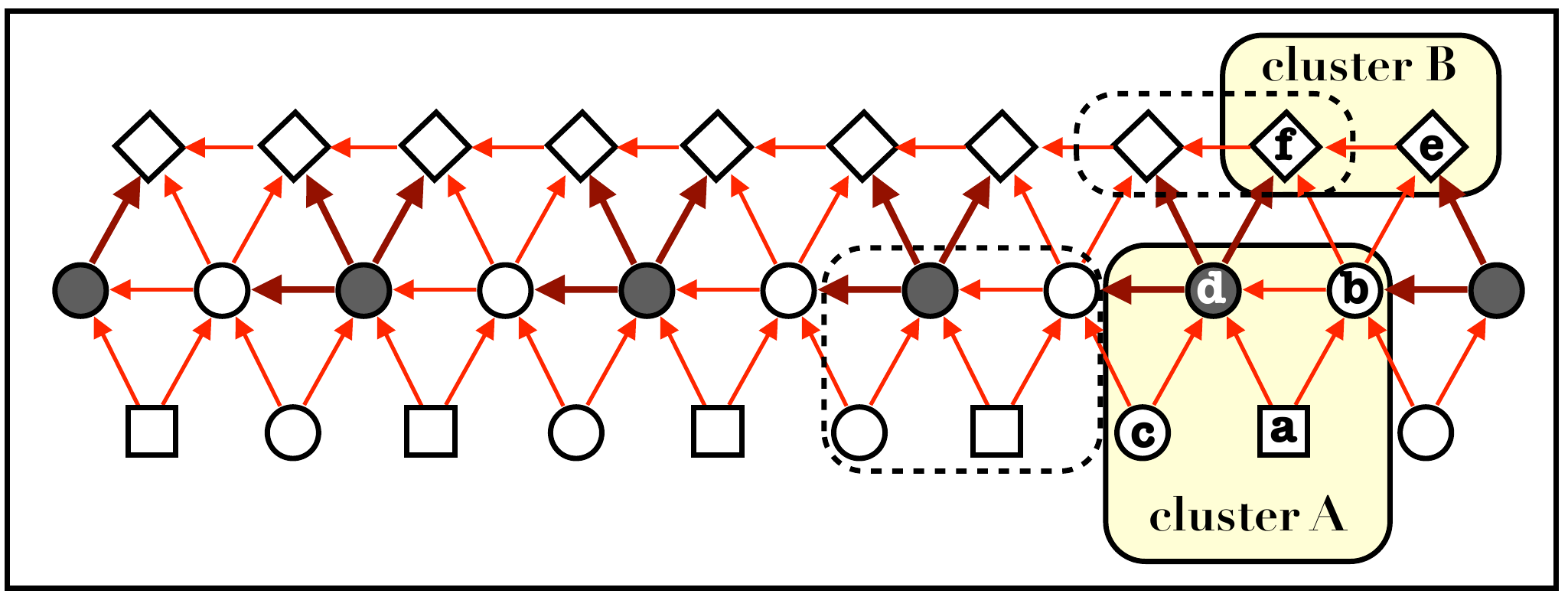}
                \caption{Interference alignment solution in a single stripe when $M=2$, and $N=4$.\vspace{-0.2in}}
                \label{iasol2x4}

\end{figure}

Assume that all the black nodes ($d_{v}=2$) have chosen their two-dimensional beamforming subspaces at random and consider the interference cluster A, shown in Fig.~\ref{iasol2x4}. 
Notice that receiver $d$ observes three interference links (one from each cell in the cluster) and receiver $b$ observers four interference streams in total (one from cell $a$ and three streams coming from the transmitters outside cluster A).

We will first show that all interference observed at receiver $b$ can be aligned in a three dimensional subspace by appropriately choosing the beamforming vector of transmitter $a$. Let $\Sm\in \CC^{4\times 3}$ denote the three-dimensional (out-of-cluster) interference observed at receiver $b$ (shown in Fig.~\ref{clustersAB}) and let $\uv_{b}\in\CC^{4\times 1}$ be a vector in the left nullspace of $\Sm$ (i.e., such that $\uv_{b}^{\rm H}\Sm=0$). The receiver $b$ can project its observation along $\uv_{b}$ to zero-force the 
out-of-cluster interference and obtain
\begin{equation}\label{rr}
\uv_{b}^{\rm H}\yv_{b} = (\uv_{b}^{\rm H}\Hm_{bb}\vv_{b})s_{b} +(\uv_{b}^{\rm H}\Hm_{ba}\vv_{a})s_{a} + \hat z. 
\end{equation}
Choosing  $\vv_{a}\in\CC^{2\times 1}$ to be orthogonal to $\Hm_{ba}^{\rm H}\uv_{b}\in\CC^{2\times 1}$ yields $\uv_{b}^{\rm H}\Hm_{ba}\vv_{a}=0$ and the receiver b can decode its message from the above interference-free observation. 
\begin{figure}[ht]

                \centering
                \includegraphics[width=.7\columnwidth]{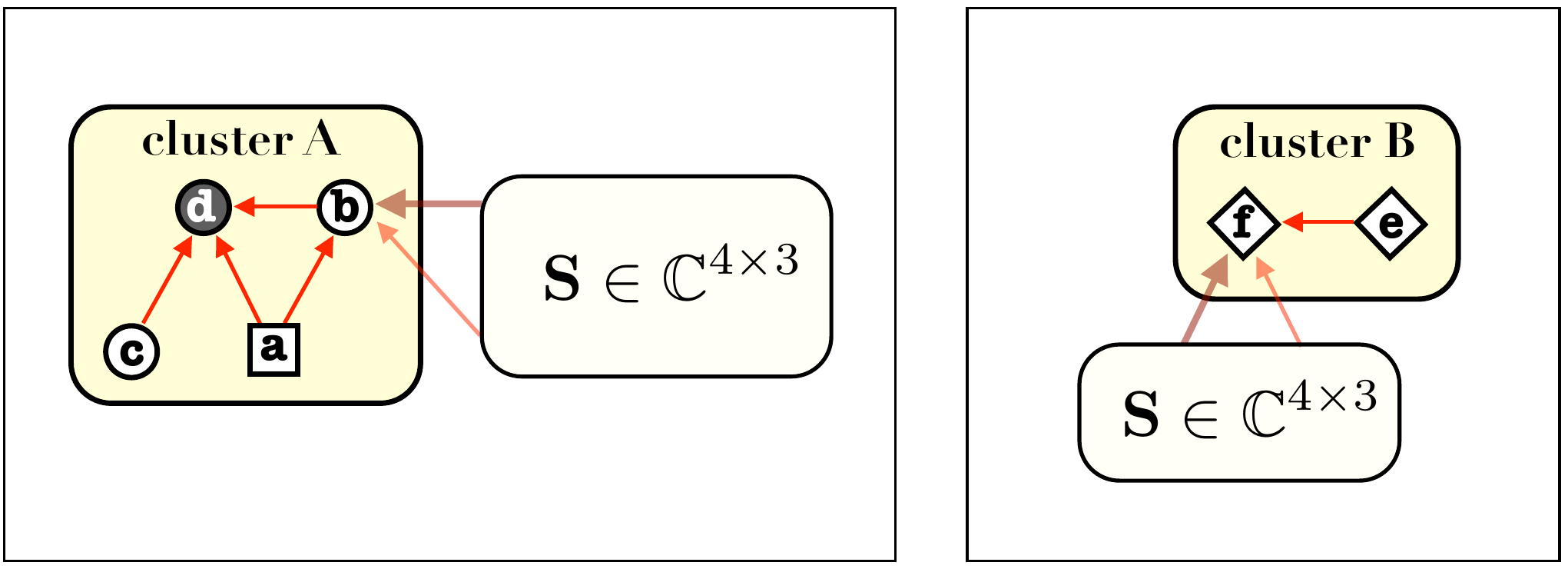}
                \caption{The interference clusters A and B.}
                \label{clustersAB}\vspace{-0.15in}

\end{figure}

Our next objective, after choosing $\vv_{a}$, is to choose $\vv_{b}$ and $\vv_{c}$ such 
that all interference is aligned in a two-dimensional subspace at receiver $d$.
The corresponding observation can be written as
\begin{align}
\yv_{d} &= \Hm_{dd}\Vm_{d}\sv_{d} + \underbrace{\Hm_{da}\vv_{a}s_{a} +\Hm_{db}\vv_{b}s_{b} +\Hm_{dc}\vv_{c}s_{c}}_{\mbox{interference}} + \zv \\
&= \Hm_{dd}\Vm_{d}\sv_{d} + \Gm\tilde\sv+ \zv, 
\end{align}
where $\Gm \triangleq {[\Hm_{da}\vv_{a}\; \Hm_{db}\vv_{b} \; \Hm_{dc}\vv_{c}]}\in \CC^{4\times 3}$ and $\tilde\sv\triangleq
[s_{a} \, s_{b} \, s_{c}]^{\rm T}
$.

\begin{lemma}
For any $\vv_{a}\neq 0$, there exist $\vv_{b}$, $\vv_{c}\in\CC^{2\times 1}$ with $||\vv_{b}||\leq 1$ and $||\vv_{c}||\leq 1$ such that ${rank}(\Gm)= 2$ with probability one, assuming that $\Hm_{da}$, $\Hm_{db}$, and $\Hm_{dc}\in \CC^{4\times 2}$ are chosen at random from a continuous (non-degenerate) distribution. 
\end{lemma}
\begin{proof}
Consider the matrix $\Fm =  [\Hm_{da}\vv_{a}\; \Hm_{db} \; \Hm_{dc}] \in \CC^{4\times 5}$ and let 
$\xv =[x_1\,x_2\,x_3\,x_4\,x_5]$ be a vector in the nullspace of $\Fm$ such that $\Fm\xv=0$.
Since $\Fm$ is generic, i.e., the corresponding matrices have been chosen from a non-degenerate distribution, we have that $x_{i}\neq 0,$ $\forall i$, with probability~one.
Let $\tilde \vv_{b} = [-x_{2}/x_{1}\;-x_{3}/x_{1}]^{\rm T}$, $\tilde \vv_{c} = [-x_{4}/x_{1}\;-x_{5}/x_{1}]^{\rm T}$, $\gamma = \min\big\{\frac{1}{||\tilde \vv_{b}||},\frac{1}{||\tilde \vv_{c}||}\big\}$ and set $\vv_{b}= \gamma\cdot\tilde \vv_{b}$ and $\vv_{c}= \gamma\cdot\tilde \vv_{c}$. Then one can check that  $\Hm_{db}\vv_{b} +\Hm_{dc}\vv_{c} = \gamma\cdot \Hm_{da}\vv_{a} $, and therefore show that ${rank}\big([\Hm_{da}\vv_{a}\; \Hm_{db}\vv_{b} \; \Hm_{dc}\vv_{c}]\big)={rank}(\Gm)= 2$. 
\end{proof}

From the above lemma we can see that  receiver $d$ will observe interference aligned in two dimensions that can be zero-forced by projecting $\yv_{d}$ along $\Um_{d}\in \CC^{4\times 2}$. The corresponding symbols $\sv_{d}$ can then be decoded from the interference-free observation
\begin{equation}
\Um_{d}^{\rm H}\yv_{d} = \Um_{d}^{\rm H}\Hm_{dd}\Vm_{d}\sv_{d} + \hat\zv.
\end{equation}

Up to this point we have shown that all the cells in cluster A (nodes $a$, $b$, $c$ and $d$) are able to decode their desired messages. Going back to Fig.~\ref{iasol2x4}, we can see that we can follow exactly the same beamforming procedure for the adjacent cluster of cells that is highlighted in a dash rectangle.
In a similar manner we can show that all the cells in the middle and bottom lines of the corresponding stripe can successfully decode their desired messages.

To conclude our proof we have to show that all the cells associated with diamond nodes (top line of the stripe) can also decode their corresponding messages. Towards this end,
consider cluster B shown in Fig.~\ref{iasol2x4} and notice that receiver $f$ observes one interfering stream from transmitter $e$ and three interfering streams from (out-of-cluster) transmitters $d$ and $b$.
Since $\Vm_{d}$ and $\vv_{b}$ have already been chosen (in cluster A), our only option is to design $\vv_{e}$ such that \begin{equation}
\Hm_{fe}\vv_{e}\in \mbox{span}\big(\underbrace{[\Hm_{fd}\Vm_{d}\;\;\Hm_{fb}\vv_{b}]}_{\triangleq\Sm\in \mathbb{C}^{4\times 3}}\big).\end{equation} 
As we have seen before (for receiver $b$), this is possible by choosing $\uv_{f}$ in the left nullspace of $\Sm$ and $\vv_{e}$ to be orthogonal to $\Hm_{fe}^{\rm H}\uv_{f}$. In a similar fashion, we can extend the above beamforming choices for all diamond nodes so that the corresponding receivers will observe all interference aligned in a three-dimensional subspace that can be zero-forced allowing them to successfully decode their desired symbols and achieve $d_{v}=1$.    
\vspace{0.1in} 
 
 \section{Conclusions}

In this work we have extended our previously proposed framework of Cellular IA to the case of 
omni-directional base-stations 
with user terminals  equipped with $M=2$ transmit antennas and we provided one-shot linear interference alignment schemes for the practically relevant cases in which base-stations  have $N=2,3$ and $4$ receive antennas. 
Omni-directional cells yield an interference graph topology that is significantly more involved than 
the previously studied sectored case. Nevertheless, omni-directional antennas are commonly used in small and densely deployed
cells, which is an on-going technology trend, and therefore deserve to be studied. 

Our schemes apply to the uplink of a cellular system, in which  base-stations are able to exchange locally decoded messages with their neighboring cells, 
but are not allowed to use received signal sharing and joint centralized decoding. This is in sharp contrast with the existing  ``distributed antenna systems'' and   
``Network MIMO'' approaches, which are much more demanding in terms of the backhaul throughput requirements. Further, it is worth mentioning that our schemes and the proposed architecture can be  implemented within the current LTE technology using the already available cellular network infrastructure.
For the case of $2\times 2$ links, we have shown that $3/4$ DoFs per user (i.e., per cell) are achievable in one-shot (without symbol extensions) by linear interference alignment precoding. Moreover, we have proven that in this case the achievable $3/4$ DoFs per user are asymptotically optimal
for a large extended network, where the boundary effect of the cells at the edge of the network vanishes.
For the practically relevant cases of $2\times 3$ and $2\times 4$ links, we have provided explicit one-shot linear IA constructions to achieve 1 and $7/6$ DoFs per user respectively.

Our converse for the $2\times 2$ case is based 
on the 
bounds 
on the degrees of freedom feasible by one-shot precoding given in \cite{Razaviyayn},\cite{Bresler}. Applying the corresponding bounds in extended networks results in a challenging non-linear optimization problem over a number of variables that grows with the number of transmit-receive pairs in the network.  Taking into account specific structural properties of the cellular interference graph, we  reformulated (relaxed) the above optimization problem into a tractable linear program with a small number of variables that does not depend on the size of the network and provided a tight converse bound for the $2\times 2$ case using duality theory.  This approach can be generalized to the symmetric $M \times M$ case, 
and is a technique of independent interest.  Overall, our work points out the importance of exploiting the network topology in 
extended (large) network models, and local decoded message sharing through backhaul links, which is not a demanding task
for the backhaul (e.g., it can be implemented by IP tunneling over a shared Fiber-Optical infrastructure) and, yet, can provide unbounded capacity gains
with respect to conventional approaches, when combined with carefully designed IA precoding schemes.

%
%

Further, we would like to comment about the nearest neighbor inter-cell interference assumption we made in this paper. 
In practice, interference from cells located further away may be relevant and one could think of settings where it should be explicitly taken into account. However,  if the effective channel coefficients of the scheme proposed in this paper (after alignment and cancellation) satisfy the TIN optimality conditions in 
\cite{gnaj13}, treating non-neighbor inter-cell interference as noise yields approximately optimal rates. 
Even when this is not the case, there is always the option to partition the network into coarser subnetworks using classical frequency 
reuse schemes, such that the TIN optimality conditions hold for each subnetwork.

Finally, it is worth pointing out that the DoFs gains obtained by ``Cellular IA'' in the uplink can also be obtained in a ``dual'' framework 
for the downlink of cellular networks with the same backhaul architecture.  In particular, in \cite{nmc14duality} we showed that for every one-shot IA 
scheme that can achieve $d$ DoFs per user in the uplink, there exists a ``dual''  Cellular IA scheme for the downlink that achieves the same DoFs.  
The dual scheme is based on dirty-paper coding \cite{costaDPC}, and the base station transmitters 
share quantized versions of their dirty-paper precoded signal with their neighboring base stations, with the same
backhaul requirements both in terms of local communication and in terms of per-link rate. The details of this duality and 
the Cellular IA schemes for the downlink are deferred to a future paper.

Altogether, we believe that the proposed Cellular IA framework can provide valuable engineering insights towards realizing the potential gains of interference alignment in current cellular technology, and lead to the design of efficient transmission schemes that meet the increasing bandwidth demands 
of  today's wireless networks.

\clearpage
\appendices
\section{Proof of Lemma \ref{lem:objective}}\label{proof:objective}
From the definition of $n_{v}$ in (\ref{nv}), we have that
\begin{equation}
\sum_{[a,b,c]\in \cal T}\left({d_{a}+d_{b}+d_{c}}\right) = \sum_{v\in \cal V}n_{v}d_{v}.
\end{equation}
Splitting the sum in terms of $\cal V_{\rm in}$ and $\cal V_{\rm ex}$ we get
\begin{eqnarray}
\sum_{v\in \cal V}n_{v}d_{v} &=&\sum_{v\in \cal V_{\rm in}}n_{v}d_{v} + \sum_{u\in \cal V_{\rm ex}}n_{u}d_{u}\\
&=& 3\sum_{v\in \cal V}d_{v} + \sum_{v\in \cal V_{\rm in}}(n_{v}-3)d_{v} + \sum_{u\in \cal V_{\rm ex}}(n_{u}-3)d_{u} \\
&=& 3\sum_{v\in \cal V}d_{v} + \sum_{v\in \cal V_{\rm ex}}(n_{v}-3)d_{v}\,,
\end{eqnarray}
where the last step follows from the fact that $n_{v} =3$ for all $v \in \cal V_{\rm in}$. Rearranging the terms and dividing by $3$ gives
\begin{equation}
\sum_{v\in\cal V}d_{v} = \hspace{-0.1in} \sum_{[a,b,c]\in \cal T}\left(\frac{d_{a}+d_{b}+d_{c}}{3}\right) + \sum_{u\in {\cal V}_{\rm ex}}\left(1-\frac{n_{u}}{3}\right)d_{u},\label{lem1}
\end{equation}
and hence , we can rewrite the average DoFs  as 
\begin{equation}
\frac{1}{|{\cal V}|} \sum_{v\in \cal V}d_{v} = \frac{1}{3|{\cal V}|}\sum_{[a,b,c]\in \cal T}
\hspace{-0.09in}(d_{a}+d_{b}+d_{c}) + \frac{C_{\rm ex}}{|\cal V|}\, ,
\label{eq24}
\end{equation}  
where 
 $$C_{\rm ex} = \sum_{v\in {\cal V}_{\rm ex}}\left(1-\frac{n_{v}}{3}\right)d_{v}.$$
 
Since $0\leq n_{v}<3$ and $0\leq d_{v}\leq M$, $\forall v\in \cal V$ we have that 

\begin{equation}C_{\rm ex}\leq \sum_{v\in \cal V_{\rm ex}}d_{v}\leq {M|{\cal{V}_{\rm ex}}|},\end{equation} and hence
\begin{align}
  \frac{1}{|{\cal V}|} \sum_{v\in \cal V}d_{v} \leq  \frac{1}{3|{\cal V}|}\sum_{[a,b,c]\in \cal T}
(d_{a}+d_{b}+d_{c})  + \frac{M|{\cal{V}_{\rm ex}}|}{|\cal V|}.
  \end{align}

\section{Proof of Lemma \ref{lem:constraint}}\label{proof:constraint}

First, observe that each triangle $[a,b,c]\in \cal T$ uniquely covers the three corresponding edges in  ${\cal E}_{\pi}$, and since 
$d_{v}\geq 0$ for all $v \in \cal V$ we can lower bound the right-hand side of (\ref{IAfeasibility3}) by $$\displaystyle \sum_{[u,v]\in\cal E_{\pi}}d_{u}d_{v} \geq  \sum_{[a,b,c]\in \cal T}\left({d_{a}d_{b}+d_{a}d_{c}+d_{b}d_{c}}\right).$$ Then, using (\ref{lem1}), we can rewrite the left-hand side of (\ref{IAfeasibility3}) as 
 \begin{align*}
 2\sum_{v\in \cal V}(M - d_{v})d_{v} &= \sum_{[a,b,c]\in \cal T}
 \frac{2}{3}\Big((M - d_{a})d_{a} + (M - d_{b})d_{b} +(M - d_{c})d_{c}\Big) \\
 &\hspace{0.5in}+ 2\sum_{v\in {\cal V}_{\rm ex}}\left(1-\frac{n_{v}}{3}\right)(M - d_{v})d_{v}\\
 &=\frac{1}{3}\sum_{[a,b,c]\in \cal T}\Big(2M(d_{a}+d_{b}+d_{c}) - 2(d_{a}^{2}+d_{b}^{2}+d_{c}^{2}) \Big)\\
 &\hspace{0.5in}+ \frac{2}{3}\sum_{v\in {\cal V}_{\rm ex}}\left(3-{n_{v}}\right)(M - d_{v})d_{v}. 
 \end{align*}
 Putting everything together, we can argue that any degrees of freedom $\{d_{v},v\in \cal V\}$ that satisfy (\ref{IAfeasibility3}) must also satisfy: 
 \begin{align}
 &\sum_{[a,b,c]\in \cal T}\big(  2(d_{a}^{2}+d_{b}^{2}+d_{c}^{2}) +3\left({d_{a}d_{b}+d_{a}d_{c}+d_{b}d_{c}}\right) -2M(d_{a}+d_{b}+d_{c})\big)
  \leq D_{\rm ex} \Leftrightarrow \nonumber\\
  &\sum_{[a,b,c]\in \cal T}\big(  
  \underbrace{(d_{a}+d_{b})^{2}+(d_{a}+d_{c})^{2}+(d_{b}+d_{c})^{2} +{d_{a}d_{b}+d_{a}d_{c}+d_{b}d_{c}} -2M(d_{a}+d_{b}+d_{c})}_{\triangleq g(d_{a},d_{b},d_{c})}\big)
  \leq D_{\rm ex}, 
  \end{align}
  where $$D_{\rm ex} = 2\sum_{v\in {\cal V}_{\rm ex}}\left(3-{n_{v}}\right)(M - d_{v})d_{v}.$$
  Since $0\leq n_{v}<3$ and $0\leq d_{v}\leq M$, $\forall v\in \cal V$ we have that
  $$\;D_{\rm ex}\leq\sum_{v\in {\cal V}_{\rm ex}}6(M - d_{v})d_{v}\leq \frac{3M^{2}}{2}|{\cal{V}_{\rm ex}}|,$$ and hence 
\begin{align}
\sum_{[a,b,c]\in \cal T}g(d_{a},d_{b},d_{c})
  \leq \frac{3M^{2}}{2}|{\cal{V}_{\rm ex}}|.
  \end{align}

\section{Proof of Lemma~\ref{lem:LP-abc}}\label{proof:lem:LP-abc}

Consider the optimization problem $LP(\sv,\gv,\alpha,\beta,\gamma)$ given by
\begin{align}
\underset{\left\{ \substack{\xv \in \mathbb{R}^{n}} \right\}}{\mbox{maximize}}\;\;\;\;\;  &\alpha\cdot \sv^{\rm T}\xv + \beta \\
\;\mbox{subject to:} \;\;\;\;\;& \gv^{\rm T}\xv \leq \gamma,\;\,{\bf 1}^{\rm T}\xv  = 1,\;\,\xv  \geq  0,  
\end{align}
where $\alpha,\beta,\gamma\geq0$ and $\sv,\gv\in \RR^{n}$, and let ${\rm opt}(\sv,\gv,\alpha,\beta,\gamma)$ denote its optimal value. The Langrangian associated with $LP(\sv,\gv,\alpha,\beta,\gamma)$ is given by 
\begin{align}
L(\xv,\lambda,\mu,\boldsymbol{\nu}) &= \alpha\cdot\sv^{\rm T}\xv + \beta - \lambda(\gv^{\rm T}\xv-\gamma)- \mu({\bf 1}^{\rm T}\xv-1) + \boldsymbol{\nu}^{\rm T}\xv \\ 
&= (\alpha\sv-\lambda\gv-\mu{\bf 1}+\boldsymbol{\nu})^{\rm T}\xv + \lambda\gamma + \mu +\beta \label{L}
\end{align}
where $\lambda$, $\mu\in \RR$ and $\boldsymbol{\nu}\in\RR^{n}$ are the Lagrange multipliers associated with the constraints $\gv^{\rm T}\xv \leq \gamma$, ${\bf 1}^{\rm T}\xv  = 1$ and $\xv  \geq  0$.
From duality theory we have that the Lagrange dual function defined as 
\begin{equation}\label{dualf}
h(\lambda,\mu,\boldsymbol{\nu}) \triangleq \sup_{\xv \in \mathbb{R}^{n}}L(\xv,\lambda,\mu,\boldsymbol{\nu})
\end{equation}
satisfies
\begin{equation}
{\rm opt}(\sv,\gv,\alpha,\beta,\gamma)\leq h(\lambda,\mu,\boldsymbol{\nu})
\end{equation}
for any $\lambda\geq 0$, $\boldsymbol{\nu}\geq 0$ and  $\mu \in \RR$.
From (\ref{L}) and (\ref{dualf}) we obtain that 
\begin{equation}
h(\lambda,\mu,\boldsymbol{\nu}) = \begin{cases}
\lambda\gamma+\mu+\beta \;\;\;\mbox{when}\;\;\;\alpha\sv-\lambda\gv-\mu{\bf 1}+\boldsymbol{\nu}=0\\
\infty \;\;\;\;\;\;\;\;\;\;\;\;\;\;\;\;\;\mbox{otherwise.} 
\end{cases}
\end{equation}
Since $\boldsymbol{\nu}\geq 0$ we have that $\alpha\sv-\lambda\gv-\mu{\bf 1}+\boldsymbol{\nu}=0 \Leftrightarrow \mu{\bf 1}\geq \alpha\sv-\lambda\gv$ and therefore setting 
$\mu^{*} =  \max_{i}\{\alpha s_{i}-\lambda g_{i}\}$ we get
\begin{equation}
{\rm opt}(\sv,\gv,\alpha,\beta,\gamma)\leq h(\lambda,\mu^{*},\boldsymbol{\nu}) = 
\max_{i}\{\alpha s_{i}-\lambda g_{i}\} +\lambda\gamma+\beta,
\end{equation}
for all $\lambda\geq 0$. Rewriting the above bound as 
$\alpha\cdot \max_{i}\{ s_{i}-\tilde\lambda\cdot g_{i}\} +\tilde\lambda{\alpha}\gamma+\beta  $
where $\tilde\lambda\triangleq\lambda/\alpha\geq 0$ we obtain the desired result.

\clearpage
\section{Proof of Lemma~\ref{construction1}}\label{proof:construction}
Recall that the set of vertices $\cal V$  in $\cal G(V,E)$ is defined in terms of a parameter $r\geq 1$ as 
\begin{equation*}
{\cal V} = \left\{ \phi^{-1}(z) : z\in \mathbb{Z}(\omega)\cap{\cal B}_{r}\right\},
\end{equation*}
where
$${\cal B}_{r} \triangleq \left\{z\in \mathbb{C}:|{\rm Re}(z)|\leq r , |{\rm Im}(z)|\leq \frac{\sqrt{3}r}{2}\right\}.$$
Since the size of the graph depends on the choice of  $r$, we will consider here the sequence of graphs ${\cal G}^{(r)}({\cal V}^{(r)},{\cal E}^{(r)})$,  indexed by $r \in \ZZ^{+}$ and provide the corresponding results in terms of the above parameter.

\subsection{The cardinality of ${\cal V}^{(r)}$}

By definition $|{\cal V}^{(r)}| =  |\mathbb{Z}(\omega)\cap{\cal B}_{r}|$. Hence, our goal is to count the number of Eisenstein integers that belong to the set $ \mathbb{Z}(\omega)\cap{\cal B}_{r}$. 
We define the sets
\begin{equation}
L(k) = \left\{z\in\mathbb{Z}(\omega)\cap{\cal B}_{r} : |{\rm Im}(z)| = \frac{\sqrt{3}k}{2}\right\}\end{equation}
for all $k\in\{-r,...,0,...,r\}$. Notice that the sets $L(k)$ contain all the Eisenstein integers that  lie on the same horizontal line on the complex plane and hence $\bigcup_{k}L(k)$  forms a partition of the set $\mathbb{Z}(\omega)\cap{\cal B}_{r}$. Therefore,
$$|\mathbb{Z}(\omega)\cap{\cal B}_{r}| = \sum_{k=-r}^{r}|L(k)|.$$
A key observation coming  from the triangular structure of $\ZZ(\omega)$ is that 
\begin{equation*}
|L(k)| = \begin{cases}|L(0)|, \;\;k \mbox{ is even}\\|L(1)|, \;\;k \mbox{ is odd.} \end{cases}
\end{equation*}
Hence, we can write
$$|\mathbb{Z}(\omega)\cap{\cal B}_{r}| = K_{\rm even}^{[r]}|L(0)| + K_{\rm odd}^{[r]}|L(1)|.$$
where $K_{\rm even}^{[r]},K_{\rm odd}^{[r]}$ denote the cardinalities of even and odd integers in $\{-r,...,0,...,r\}$. 

If $r$ is even then $K_{\rm even}^{[r]} = r+1$ and $K_{\rm even}^{[r]} = r$, whereas if $r$ is odd then $K_{\rm even}^{[r]} = r$ and $K_{\rm even}^{[r]} = r+1$. Since $|L(0)|=2r+1$ and $|L(1)|=2r$ for all $r\geq 1$ we have that
\begin{equation}
|{\cal V}^{(r)}| = |\mathbb{Z}(\omega)\cap{\cal B}_{r}| = \begin{cases}4r^{2}+3r+1, \;\;r \mbox{ is even}\\4r^{2}+3r, \;\;\;\;\;\;\;\;\,r \mbox{ is odd.} \end{cases}
\label{cardV}
\end{equation}

\subsection{The cardinality of ${\cal T}^{(r)}$}

We will associate here each ordered vertex triplet $[a,b,c]\in {\cal T}^{(r)}$ with its leading vertex $a\in {\cal V}^{(r)}$ in a one-to-one fashion and define the set
$${\cal A}^{(r)} =\{\phi^{-1}(u)\in \mathbb{Z}(\omega)\cap{\cal B}_{r} : [a,b,c]\in {\cal T}^{(r)}\}.$$
In order to determine the cardinality of ${\cal T}^{(r)}$,  it suffices to count the number of Eisenstein integers that belong to the set ${\cal A}^{(r)}$, since $|{\cal T}^{(r)}| =|{\cal A}^{(r)}|$ by definition.
Consider the sets

$$S(k)= {\cal A}^{(r)}\cap L(k)$$
for all $k\in\{-r,...,0,...,r-1\}$. The set $S(k)$ contains all the Eisenstein integers that are associated with a leading vertex of a triangle and lie on the same horizontal line $L(k)$.  
As before, $\bigcup_{k}S(k)$
forms a partition of ${\cal A}^{(r)}$ and hence
$$|{\cal A}^{(r)}| = \sum_{k=-r}^{r-1} |S(k)|.$$
Intuitively $|S(k)|$ counts the number of triangles that are formed between the lines $L(k)$ and $L(k+1)$ and hence the total number of triangles can be obtained by adding all $|S(k)|$ up to $k=r-1$.

It is not hard to verify that 
\begin{equation*}
|S(k)| = \begin{cases}|S(0)|, \;\;k \mbox{ is even}\\|S(1)|, \;\;k \mbox{ is odd,} \end{cases}
\end{equation*}
for all $r\geq 2$ and hence 
$$|{\cal A}^{(r)}|  = K_{\rm even}^{[r]}|S(0)| + K_{\rm odd}^{[r]}|S(1)|$$
where $K_{\rm even}^{[r]},K_{\rm odd}^{[r]}$ denote the cardinalities of even and odd integers in $\{-r,...,0,...,r-1\}$. We have that $K_{\rm even}^{[r]}=K_{\rm odd}^{[r]}=r$ and hence
$$|{\cal A}^{(r)}|  = r\left(|S(0)| + |S(1)|\right).$$
It follows from the definitions of ${\cal T}^{(r)}$, ${\cal A}^{(r)}$ and $S(0)$ that  
\begin{equation*}
z\in S(0) \Leftrightarrow  z\in L(0)\; \mbox{and}\; z+\omega,z+\omega+1 \in L(1).
\end{equation*}
We can argue hence that the set $S(0)$ hence contains the integers $a \in \{-r+1,...,r-1\}$.

Similarly,
\begin{equation*}
z\in S(1) \Leftrightarrow  z\in L(1)\; \mbox{and} \;z+\omega,z+\omega+1 \in L(2),
\end{equation*}
and hence the set $S(1)$ contains the Eisenstein integers $z=a + \omega$, for all $a\in \{-r+1,...,r\}$.

It follows that $|S(0)| =2r-1$ and $|S(1)| =2r$
and therefore 
\begin{align}
|{\cal T}^{(r)}| =|{\cal A}^{(r)}|= 4r^{2}-r. 
\label{cardT}
\end{align}

Comparing (\ref{cardT}) with (\ref{cardV}) gives $|{\cal T}^{(r)}|\leq|{\cal V}^{(r)}|$ as stated in the first part of the lemma.

\subsection{The cardinality of ${\cal V}_{\rm ex}^{(r)}$}

We will upper bound $|{\cal V}_{\rm ex}^{(r)}|$ as follows.
From Lemma~\ref{lem1} we have that
\begin{equation*}
\sum_{u\in {\cal V}_{\rm ex}^{(r)}}\hspace{-0.1in}\left(1-\frac{n_{u}}{3}\right)x_{u}\hspace{-0.05in} = \hspace{-0.1in}\sum_{v\in{\cal V}^{(r)}}x_{v} - \hspace{-0.2in} \sum_{[i,j,k]\in {\cal T}^{(r)}}\hspace{-0.1in}\left(\frac{x_{i}+x_{j}+x_{k}}{3}\right),
\end{equation*}
for any $\{x_{v}: v\in {\cal V}^{(r)}\}$. Setting $x_{v}=1,\,\forall v\in {\cal V}^{(r)}$, we obtain

\begin{equation*}
|{\cal V}_{\rm ex}^{(r)}|-\sum_{u\in {\cal V}_{\rm ex}}\frac{n_{u}}{3} = |{\cal V}^{(r)}| - |{\cal T}^{(r)}|.
\end{equation*}
Since $n_{v}\leq 2$ for all $v\in {\cal V}_{\rm ex}^{(r)}$ we have that 
$$\sum_{u\in {\cal V}_{\rm ex}}\frac{n_{u}}{3} \leq \frac{2}{3}|{\cal V}_{\rm ex}^{(r)}|, $$
and hence 
\begin{equation}
|{\cal V}_{\rm ex}^{(r)}| \leq 3\left(|{\cal V}^{(r)}| - |{\cal T}^{(r)}|\right).
\label{cardVexTMP}
\end{equation}

From (\ref{cardV}) and (\ref{cardVexTMP}) we obtain 
\begin{eqnarray}
|{\cal V}_{\rm ex}^{(r)}| &\leq&  3\left(4r^{2} + 3r + 1 - 4r^{2} +r\right) \nonumber\\
&=&12r + 3.
\label{cardVex}
\end{eqnarray}

From (\ref{cardV}) it follows that $\sqrt{|{\cal V}^{(r)}|} \geq 2r$ for all $r\geq 1$.
From (\ref{cardVex}) we have that 
$$|{\cal V}_{\rm ex}^{(r)}| \leq 12r +3 \leq 9\sqrt{|{\cal V}^{(r)}|},\; \forall r\geq 1,$$
and hence $|{\cal V}_{\rm ex}^{(r)}| = {\cal O}\left(\sqrt{|{\cal V}^{(r)}|}\right)$.

\clearpage
\section{The Linear Programming Converse for $M\times M$ Cellular Systems}\label{generalM}

Recall from Section \ref{converse} that the following linear program can be used to upper bound the average (per cell) DoFs achievable in $\cal G(V,E)$ for any decoding order $\pi$.

\begin{align*}
{{ \rm LP}({{\cal G}_{\pi}})}:\;\; \underset{\left\{ \substack{x_{(i,j,k)}\in \mathbb{R},\\ [i,j,k]\in \cal D} \right\}}{\mbox{maximize}}\;\;\;  & \frac{|{\cal T}|}{3|{\cal V}|} \sum_{[i,j,k] \in \cal \cal D}s{(i,j,k)}\cdot x_{(i,j,k)} + \frac{M|{\cal V}_{\rm ex}|}{|{\cal V}|} \\
\mbox{subject to:} \;\;& \sum_{[i,j,k] \in \cal\cal D} g{(i,j,k)}\cdot x_{(i,j,k)} \leq \frac{3M^{2}|{\cal V}_{\rm ex}|}{2|{\cal T}|}
\\
&\sum_{[i,j,k] \in \cal D} x_{(i,j,k)}  = 1 \\
& \;\;\,x_{(i,j,k)}  \geq  0, \; \forall {[i,j,k] \in \cal \cal D}   
\end{align*}

\vspace{0.1in}
Using the result of  Lemma~\ref{lem:LP-abc}, we can show that the solution of the linear program ${\rm LP}({{\cal G}_{\pi}})$ is upper bounded by

\begin{equation} \label{lpbound1}
{\rm opt}({\rm LP}({{\cal G}_{\pi}}))\leq \frac{|{\cal T}|}{3|{\cal V}|} \cdot \max_{[i,j,k]\in {\cal D}}\left\{s(i,j,k)-\frac{g(i,j,k)}{2M} \right\} + \frac{5M}{4}\cdot\frac{|{\cal V}_{\rm ex}|}{|{\cal V}|},\vspace{0.3in}
\end{equation}
by identifying $\alpha = \frac{|{\cal T}|}{3|{\cal V}|}$, $\beta=\frac{M|{\cal V}_{\rm ex}|}{|{\cal V}|}$ , $\gamma= \frac{3M^{2}|{\cal V}_{\rm ex}|}{2|{\cal T}|}$ and setting $\lambda = \frac{1}{2M}$.
Further, using Lemma~\ref{construction1} we obtain 

\begin{equation}
{\rm opt}({\rm LP}({{\cal G}_{\pi}}))\leq \frac{1}{3} \cdot \max_{[i,j,k]\in {\cal T}_{\rm D}}\left\{s(i,j,k)-\frac{g(i,j,k)}{2M} \right\} + {\cal O}\left(\frac{1}{\sqrt{|{\cal V}|}}\right).
\end{equation}
\vspace{0.1in}

Putting everything together and defining the function 
\begin{align}
f_{M}{(i,j,k)}&\triangleq \frac{1}{3}\left(s(i,j,k)-\frac{g(i,j,k)}{2M}\right)\\
&= \frac{1}{3}\left(2(i+j+k) - \frac{(i+j)^{2}+ij +(i+k)^{2}+ik+(j+k)^{2}+jk}{2M}\right),
\label{FM}
\end{align}
we arrive at the following theorem that upper bounds the average (per cell) DoFs in our framework.

\begin{thm}  \label{thm3}
For a cellular system $\cal G(V,E)$ in which transmitters and receivers are equipped with $M$ antennas each and for any network interference 
cancellation decoding order $\pi$,  the average (per cell) DoFs   that can be achieved by any one-shot linear beamforming scheme  are bounded by 
$$\displaystyle
 d_{{\cal G},{\pi}} \leq \max_{[i,j,k]\in {\cal D}}f_M{(i,j,k)} + {\cal O}\left( \scriptstyle{1}/{{\sqrt{|{\cal V}|}}}\right),$$
 where ${\cal D}$ is the set of all distinct triangle-DoF configurations given in (\ref{TD}).
\hfill \QED
\end{thm}

\begin{table}[h]
    \caption{ The sets ${\cal T}_{\rm D}$ and corresponding values of $f_M{(i,j,k)}$ for $M=2,3$ and $4$.
The asterisks indicate the configurations $[i,j,k]\in \cal D$ that attain the maximum of $f_M{(i,j,k)}$.}
    \begin{minipage}{\linewidth}
      
      \centering
      \begin{tabular}[t]{ |c|c| }
  \hline
  \multicolumn{2}{|c|}{$M=2$} \\
  \hline\hline
  ${\cal D}$ & $f_M{(i,j,k)}$\\
  \hline
  
  $[0,0,0]$ & $0$ \\
  $[0,0,1]$ & ${1}/{2}$ \\
  $[0,0,2]$ & ${2}/{3}$ \\
  $\bf[0,1,1]$ & $\bf{3}/{4}^{*}$ \\
  $\bf[1,1,1]$ & $\bf{3}/{4}^{*}$ \\
  \hline
\end{tabular}
\vspace{1in}
\quad
\begin{tabular}[t]{ |c|c| }
  \hline
  \multicolumn{2}{|c|}{$M=3$} \\
  \hline\hline
  ${\cal D}$ & $f_M{(i,j,k)}$\\
  \hline
  
  $[0,0,0]$ & $0$ \\
  $[0,0,1]$ & ${5}/{9}$ \\
  $[0,0,2]$ & ${8}/{9}$ \\
  $[0,0,3]$ & $1$ \\
  $[0,1,1]$ & ${17}/{18}$ \\
  $[0,1,2]$ & ${10}/{9}$ \\
  $\bf[1,1,1]$ & $\bf{7}/{6}^{*}$ \\
  $\bf[1,1,2]$ & $\bf{7}/{6}^{*}$ \\
  \hline
\end{tabular}
\quad
\begin{tabular}[t]{ |c|c| }
  \hline
  \multicolumn{2}{|c|}{$M=4$} \\
  \hline\hline
  ${\cal  D}$ & $f_M{(i,j,k)}$\\
  \hline
  
  $[0,0,0]$ & $0$ \\
  $[0,0,1]$ & ${7}/{12}$ \\
  $[0,0,2]$ & $1$ \\
  $[0,0,3]$ & $5/4$ \\
  $[0,0,4]$ & $4/3$ \\
  $[0,1,1]$ & ${25}/{24}$ \\
  $[0,1,2]$ & ${4}/{3}$ \\
  $[0,1,3]$ & ${35}/{24}$ \\
  $[0,2,2]$ & ${3}/{2}$ \\
  $[1,1,1]$ & ${11}/{8}$ \\
  $[1,1,2]$ & ${37}/{24}$ \\
  $[1,1,3]$ & ${37/24}$ \\
  $\bf[1,2,2]$ & $\bf{19}/{12}^{*}$ \\
  $[2,2,2]$ & ${3}/{2}$ \\
  \hline
\end{tabular}
   \end{minipage} 
  \vspace{-1in}
\end{table}

\begin{cor}We have that $\displaystyle \max_{[i,j,k]\in {\cal T}_{\rm D}}f_M{(i,j,k)}\leq \frac{2M}{5}, \;\;\mbox{for all}\; M.$
\end{cor}
\begin{proof}The function $f_M{(i,j,k)}$ is  concave  with global maximum $\displaystyle\max_{[i,j,k]\in \mathbb{R}^{3}}f_M{(i,j,k)} = {2M}/{5}$ at $[i,j,k] = \left[\frac{2M}{5},\frac{2M}{5},\frac{2M}{5}\right]$. 
Since $\displaystyle\max_{[i,j,k]\in {\cal T}_{\rm D}}f_M{(i,j,k)}\leq \max_{[i,j,k]\in \mathbb{R}^{3}}f_M{(i,j,k)}$ for all $M$, the corollary follows.
\end{proof}
\clearpage

\bibliographystyle{ieeetr}
\bibliography{referencesIT}

\end{document}

%% file: macros.tex
\long\def\comment#1{}

\newfont{\bbb}{msbm10 scaled 800}

\newfont{\bb}{msbm10 scaled 1100}
\newcommand{\CC}{\mbox{\bb C}}

\newcommand{\RR}{\mbox{\bb R}}

\newcommand{\ZZ}{\mbox{\bb Z}}

\newcommand{\gv}{{\bf g}}

\newcommand{\sv}{{\bf s}}

\newcommand{\uv}{{\bf u}}

\newcommand{\vv}{{\bf v}}
\newcommand{\xv}{{\bf x}}
\newcommand{\yv}{{\bf y}}
\newcommand{\zv}{{\bf z}}

\newcommand{\Fm}{{\bf F}}
\newcommand{\Gm}{{\bf G}}
\newcommand{\Hm}{{\bf H}}

\newcommand{\Pm}{{\bf P}}

\newcommand{\Sm}{{\bf S}}

\newcommand{\Um}{{\bf U}}

\newcommand{\Vm}{{\bf V}}

\newcommand{\Ec}{{\cal E}}

\renewcommand{\arg}{{\hbox{arg}}}
